\documentclass[12pt]{article}
\usepackage{amsxtra,amssymb,amsmath,graphicx,bbm,amsthm,color,mathrsfs,enumerate}
\usepackage[urlcolor=black,linkcolor=black,citecolor=black,colorlinks=true]{hyperref}
\usepackage[active]{srcltx}
\usepackage{t1enc}
\usepackage{aurical}
\usepackage[T1]{fontenc}
\usepackage{comment}
\usepackage[left=2.5cm, right=2.5cm, top=2cm, bottom=4cm]{geometry}
\usepackage{caption}  %captionof

\newcommand{\e}{\mathbb{E}}
\newcommand{\E}{\mathbb{E}}

\renewcommand{\P}{\mathbb{P}}
\newcommand{\R}{\mathbb{R}}

\newcommand{\1}{{\mathbf 1}}

\def\be{\begin{align}}
\def\ee{\end{align}}
\def\b*{\begin{eqnarray*}}
\def\e*{\end{eqnarray*}}

\def\vp{\varphi}

\def\be{\begin{eqnarray}}
\def\ee{\end{eqnarray}}
\def\beq{\begin{equation}}
\def\eeq{\end{equation}}
\def\b*{\begin{eqnarray*}}
\def\e*{\end{eqnarray*}}
\def\bi{\begin{itemize}}
\def\ei{\end{itemize}}

\def \1{{\bf 1}}
\def\vp{\varphi}
\def\eps{\varepsilon}

\def\={\;=\;}

\def\x{\times}

\def \proof{{\noindent \bf Proof. }}
\def \ep{\hbox{ }\hfill$\Box$}
 \def\reff#1{{\rm(\ref{#1})}}

 \def\vs#1{\vspace{#1mm}}

\def \E{\mathbb{E}}
\def \F{\mathbb{F}}

\def \N{\mathbb{N}}
\def \P{\mathbb{P}}

\def \R{\mathbb{R}}
\def\T{\mathbb{T}}

\def\Bc{{\cal B}}

\def\Ec{{\cal E}}
\def\Fc{{\cal F}}

\def\Lc{{\cal L}}

\def\Nc{{\cal N}}
\def\Kc{{\cal K}}

\def\Lc{{\cal L}}

\def\Xb{\bar X}

\def\Zb{\bar Z}

\newtheorem{Theorem}{Theorem}[part]

\newtheorem{Proposition}{Proposition}[part]
\newtheorem{Assumption}{Assumption}[part]
\newtheorem{Lemma}{Lemma}[part]

\newtheorem{Remark}{Remark}[part]

\makeatletter \@addtoreset{equation}{section}

\@addtoreset{Definition}{section}

\@addtoreset{Theorem}{section}

\@addtoreset{Proposition}{section}

\@addtoreset{Assumption}{section}

\@addtoreset{Corollary}{section}

\@addtoreset{Lemma}{section}

\@addtoreset{Remark}{section}

\@addtoreset{Example}{section}

\def\qr{{\rm q}}
\def\Qr{{\rm Q}}
\def\kr{{\rm k}}
\def\kr{{\rm k}}
\def\Xb{\mathbf{X}}
\def\Mb{\mathbf{M}}
\def\vr{{\rm v}}
\def\Kc{{\mathcal K}}
\def\Zb{\mathbf Z}
\def\as{\rm a.s.}
\def\T{{\rm T}}

\def\zr{{\rm z}}

\def\U{{\rm U}}
\def\Ab{\mathbf A}
\def\Ec{{\rm E}}

\begin{document}

\title{Optimal trading with online parameters revisions}

\author{N. Baradel\thanks{ENSAE-ParisTech, CREST, and, Université Paris-Dauphine, PSL Research University, CNRS, UMR [7534], CEREMADE, 75016 Paris, France. }, B. Bouchard\thanks{Université Paris-Dauphine, PSL Research University, CNRS, UMR [7534], CEREMADE, 75016 Paris, France. This research is supported by the Initiative de Recherche  ``Stratégies de Trading et d'Investissement Quantitatif'',   Kepler-Chevreux and Collège de France.}, N.~M. Dang\thanks{Kepler-Chevreux 112 Avenue Kléber, 75116 Paris, France.}}
\date{\today}
\maketitle

\begin{abstract} The aim of this paper is to explain how  parameters adjustments can be integrated in the design or the control of  automates of trading. Typically, we are interested by the online estimation of the market impacts generated by robots or single orders, and how they/the controller should react in an optimal way to the informations  generated by the observation of the realized impacts. This can be formulated as an optimal impulse control problem with unknown parameters, on which a prior is given. 
We explain how  a mix of the classical Bayesian updating rule and of optimal control techniques allows one to derive the dynamic programming equation satisfied by the corresponding value function, from which the optimal policy can be inferred. We provide an example of convergent finite difference scheme and consider typical examples of applications. 
  \end{abstract}
 

\parindent=0pt

\section{Introduction}

 The design of trading algorithms is based on two corner stones: first a model need {to} be estimated, second an optimal trading policy has to be computed, given a prescribed criteria. 
 
As in traditional portfolio management, the notions of volatility or correlation play an important role and need to be estimated. At the level of an order book, one can also be interested by the speed of arrival/cancelation of orders, etc.  These refer to the concept of exogenous dynamics, which can be inferred offline, see e.g.~\cite{bacry2014hawkes,delattre2013estimating,rosenbaum2007etude} for references, and possibly adjusted by using traditional filtering techniques, see e.g.~\cite{bensoussan2004stochastic,duncan2002adaptive}.  
 
More importantly, a trading robot has his own impact on the dynamics of the traded assets, either because of market microstructure effects at the high frequency level, or because traded volumes are non-neglectable with respect to the so-called market volume, see \cite{bouchaud2010price,lehallebook} for surveys as well as \cite{brokmann2014slow} for recent references. The knowledge of this impact is crucial. However, unlike  other market parameters, it can be observed only when the algorithm is actually running, offline estimation cannot be done.

As a matter of fact, the controller faces a classical dilemma: should he impulse the system (pass orders) to  gain immediately more information to the price of possible immediate losses, or rather try to maximise his current expected reward with the risk of not learning for the future?   In any case, the optimal trading policy has to incorporate the fact that the knowledge on the impact parameters will evolve along  time, as the effects of the robot on the system are revealed, and that the uncertainty on their value is a source of risk.     

One way to analyse this situation  is to use the multi-arms bandit recursive learning  approach of \cite{laruelle2011optimal,laruelle2013optimal}. By sending successive impulses to the system, one increases his knowledge on the true distribution of the response function.  It provides asymptotically optimal policies. The main advantage is that it is model free. On the other hand it requires (weak) ergodicity-type conditions which have little chance to be satisfied if the price/book order is actually impacted. Moreover, the global flow of orders is not optimal, {it} starts to be optimal only in the long range. 
\vs2

In this paper, we propose to use the classical Bayesian approach, see e.g.~\cite{easley} for general references. A  similar idea has already been suggested in \cite{almgren2006bayesian} in the context of  optimal trading, in a different and very particular framework (trend estimation). In this approach, one fixes an {\sl a priori}  distribution on the true value of the parameters. When a new post-trade information arrives, this distribution is updated according to the  Bayesian rule to provide an {\sl a posteriori} distribution, which will be used as a new prior for the next (bunch of) trade(s). This can, in theory, address pretty general model-free optimal control problems, in the sense that the unknown parameter can be the whole distribution of the response function. Still, the {\sl a posteriori} distribution will be dominated by the prior. If the true distribution is not, the sequence of estimators cannot converge to it. Also, in practice, computational time constraints will force one to restrict to a class of parameterized distributions, so as to reduce to a space of small dimension. Hence, it requires to choose a set of possible models. Our procedure will only reveal what is the value of the parameters most likely to be, given this prescribed framework.  In practice,  classes of models are already used by practitioners, but parameters are difficult to estimate at a large scale in an automatic way and evolve according to time and market conditions. From this perspective, we believe that our approach allows us to address their estimation issues in an efficient way. 

There are two ways to consider this updating mechanism. One is that it allows one to estimate the true value of the parameters while acting in an optimal manner. This is the point of view developed in \cite{easley}, who provides conditions for a discrete time infinite horizon model that ensure the convergence of the {\sl a posteriori} law to the Dirac mass at the true parameter value. Obviously, this requires identification conditions since we can observe the effect on the system only along the optimal policy, as well as some ergodicity conditions.  Another point of view is that it actually falls out automatically from dynamic programming considerations when solving the optimal control problem associated to the initial prior. 
\vs2

We focus here on the second point of view: given a prior on the true parameter, how should we act in a optimal way? After all, we want to be optimal given a prior, and not act on the system  just to refine the prior if this is not optimal.  Obviously, if the algorithm has to be used repeatedly on the same market, and if the conditions on the market are stable, the {\sl a posteriori} distribution obtained after running the algorithm can be used as a new, more precise, prior for the next time it will be launched. The convergence issue to a Dirac distribution is left for future researches. Again, this is more of academic interest.  

\vs2

We  consider here the general abstract formalism introduced in \cite{BBD16} which aims at pertaining for most models used in practice.  In particular, we can work at the meta-order level (control of smart-routines) or at the  level of single orders (design of smart-routines or scheduling). In the first case, we are interested  by the optimal control of robots that are already given, this is similar to \cite{boucharddanglehalle}.  We only observe the global impact of the robot, after it stops. The second case concerns the design of the robots themselves, as in e.g.~\cite{avellaneda2008high,gueant2012optimal,laruelle2011optimal,laruelle2013optimal}. A typical question is whether an order should be passive or aggressive. When an aggressive order is posted,  we possibly observe an impact immediately (unless it is of small size). As for a (bunch of) limit order(s), we infer the speed at which they have been executed, if they are before being cancelled.   The choice of the trading platform can be addressed similarly, etc. 
\vs2

The rest of this paper is organized as follows. After presenting the general framework of   \cite{BBD16}, we provide their main  characterization:  the value function  is the  solution of a quasi-variational partial differential equation  from which one can infer the optimal trading policy.  We then explain how the equation can be solved numerically. This is illustrated by toy models inspired from the literature, on which simulated based optimal strategies are provided.

%%%%%%%%%%%%%%%%%%%%%%%%%%%%%%%%%%%%%%%%%%%%%%%
\section{Abstract framework}\label{sec: setup}

Before to consider typical examples of application, we present here the abstract framework proposed in the companion paper \cite{BBD16}. We believe that it is flexible enough to pertain for most problems encountered in practice. 
\vs2 

We model the driving noise by a   $d$-dimensional Brownian motion   $W$ (defined on the canonical space $C([0,T],\R^{d})$ endowed with the Wiener measure $\P$). One could also consider jump type processes, such as compound Poisson processes, the same analysis would apply.
 
The unknown parameter $\upsilon$ is supported by  a (Polish) space $(\U,\Bc(\U))$, and our initial prior on $\upsilon$ is assumed to belong to a  locally compact subset $\Mb$ of   the set of Borel probability measures on $\U$ (endowed with the topology of weak convergence).  In the applications of Section \ref{sec: examples}, the collection of possible priors can be identified as a subset of a finite dimensional space (e.g.~the parameters of a Gaussian distribution, the weights of law with finite support, etc.).  Then, $\Mb$ can be simply viewed as a finite dimensional space.

To allow for additional randomness in the measurement of the effects of trades on the system, we consider another (Polish) space $\Ec$ on which is defined a family $(\epsilon_{i})_{i \geq 0}$ of i.i.d.~random variables with common measure $\mathbb{P}_{\epsilon}$ on $\Ec$. On the product space $ {\Omega} := C([0,T],\R^{d}) \times \U \times \Ec^{\mathbb{N}}$, we consider the family of measures $\{\mathbb{P} \times m \times \mathbb{P}_{\epsilon}^{\otimes \mathbb{N}} : m \in \Mb\}$ and denote by $\mathbb{P}_{m}$ an element of this family whenever $m \in \Mb$ is fixed. The operator $\mathbb{E}_{m}$ is the expectation associated to $\mathbb{P}_{m}$. Note that $W$, $\upsilon$ and $(\epsilon_{i})_{i\ge 0}$ are independent under each $\P_{m}$. 

For $m\in \Mb$ given, we let $\F^{m}=(\Fc^{m}_{t})_{t\ge 0}$ denote the $\P_{m}$-augmentation of the filtration $ \F= (  \Fc_{t})_{t\ge 0}$ defined by $  \Fc_{t}=\sigma((W_{s})_{s\le t}, \upsilon,(\epsilon_{i})_{i \geq 0})$ for $t\ge 0$. Hereafter, all the random variables are considered with respect to the probability space $(\Omega,\Fc_{T}^{m})$ with $m\in \Mb$ given by the context, and where $T$ is a fixed time horizon.

 %%%%%%%%%%%%%%%%%%%%%%%%%%%%%%%
\subsection{The controlled system}

Let  $\Ab \subset [0, T] \times \mathbb{R}^{d}$ be  a (non-empty) compact set. It will be the set in which our controls (trading policies) will take place.   Given $N\in \N$ and $m\in \Mb$, we denote by  $\Phi^{\circ,m}_{N}$ the collection of sequences of random variables $\phi=(\tau_{i},\alpha_{i})_{i\ge 1}$ on $(\Omega,\Fc_{T}^{m})$ with values in $\R_{+}\x \Ab$ such that $(\tau_{i})_{i\ge 1}$ is a non-decreasing sequence of $\F^{m}$-stopping times satisfying $\tau_{j}>T$ $\P_{m}-\as$~for $j>N$.  We set 
    \[
        \Phi^{\circ,m} := \bigcup_{N \geq 1}\Phi^{\circ,m}_{N}.
    \]
 An element $\phi=(\tau_{i},\alpha_{i})_{1\le i\le N}\in \Phi^{\circ,m}$ will be our impulse control and we write $\alpha_{i}$ in the form 
 $$
 \alpha_{i}=(\ell_{i},\beta_{i}) \mbox{ with }  \ell_{i}\in [0,T] \mbox{ and }\beta_{i}\in \R^{d} \;\P_{m}-\as
 $$  
 More precisely, the $\tau_{i}$'s will be the times at which an impulse is made on the system (e.g.~a trading robot is launched), $\beta_{i}$ will model the nature of the order send at time $\tau_{i}$ (e.g.~the parameters used for the trading robot), and $\ell_{i}$ will stand for the maximal time length during which  no new intervention on the system can be made (e.g.~the time prescribed to the robot to send orders on the market).  
\vs2

From now on, we shall always use the notation $(\tau^{\phi}_{i},\alpha^{\phi}_{i})_{i\ge 1}$ with $\alpha^{\phi}_{i}=(\ell^{\phi}_{i},\beta^{\phi}_{i})$ to refer to a control $\phi \in \Phi^{\circ,m}$. 
\vs2
    
We allow for not observing nor being able to act on the system before  a random time $\vartheta^{\phi}_{i}$ defined by 
 \[
    \vartheta^{\phi}_{i} := \varpi(\tau^{\phi}_{i},X^{\phi}_{\tau^{\phi}_{i}-},\alpha^{\phi}_{i}, \upsilon, \epsilon_{i}),
  \]
where  $X^{\phi}$ is the controlled state process (stock prices, market volumes, wealth, etc.) that will be described below, and 
\be\label{eq: hype varpi}
\varpi :\R_{+}\x \R^{d}\x \Ab \times \U \times \Ec \rightarrow [0, T]\;\text{ is measurable, such that $\varpi(t,\cdot)\ge t$ for all $t\ge 0$. }
\ee  
In the case where the actions consist in launching a trading robot at $\tau_{i}^{\phi}$ during a certain time $\ell^{\phi}_{i}$,  we can naturally take $ \vartheta^{\phi}_{i}=\tau^{\phi}_{i}+\ell^{\phi}_{i}$. If the action consists in placing a limit order during a maximal duration $\ell^{\phi}_{i}$, $\vartheta^{\phi}_{i}$ is the time at which the limit order is executed if it is less than $\tau^{\phi}_{i}+\ell^{\phi}_{i}$, and $\tau^{\phi}_{i}+\ell^{\phi}_{i}$ otherwise. 
\vs2

We say that $\phi \in \Phi^{\circ,m}$ belongs to $\Phi^{m}$ if 
\[
\vartheta^{\phi}_{i}\le \tau^{\phi}_{i+1}\; \mbox{ and }\; \tau^{\phi}_{i}<\tau^{\phi}_{i+1}\;\;\mbox{ $\P_{m}$-a.s. for all $i\ge 1$}, 
\]
and define 
\be\label{eq: def Nphi}
\Nc^{\phi}:=\left[\cup_{i\ge 1}[\tau^{\phi}_{i},\vartheta^{\phi}_{i})\right]^{c}.
\ee

Let us now describe our controlled state process. 
Given some initial data $z:=(t, x) \in \Zb := [0,T]\x \R^{d}$,  and  $\phi \in \Phi^{m}$, we let $X^{z,\phi}$ 
be the unique strong solution on $[t,2T]$ of 
\begin{align}
		X = x&+  \left( \int_{t}^{\cdot}\1_{\Nc^{\phi}}(s)\mu\left(s,X_s\right)ds + \int_{t}^{\cdot}\1_{\Nc^{\phi}}(s)\sigma\left(s,X_{s}\right)dW_{s}   \right) \nonumber       \\
		&+  \sum_{i\ge 1}   \1_{\{t\le \vartheta^{\phi}_{i}\le \cdot\}} [ F(\tau^{\phi}_{i},X_{\tau^{\phi}_{i}-},   \alpha^{\phi}_{i}, \upsilon, \epsilon_{i})-X_{\tau^{\phi}_{i}-}]  .\label{eq: dyna X} 
\end{align}

Depending on the choice of the model, the different components of $X$ can be the {cumulative} gains of the  {algorithm}, the {number of holding shares}, {the} mid-price, the size and positions of the first bid and ask queues, the position of the trader in terms of limit orders in the different queues, a factor {driving} the system, a flow of external information provided by experts, the current market volume, etc. This can also be the time itself, if one wants to have time dependent dynamics.  This is quite flexible, and we will exemplify this in  Section \ref{sec: examples}.  
\vs2

In the above, the function 
\be\label{eq: hyp mu sigma F}
\begin{array}{c}
(\mu,\sigma,F) :  \R_{+}\times \R^{d} \x \Ab \times \U \times \Ec\mapsto  \R^{d}\times \mathbb{R}^{d\times d} \times \R^{d} \;\text{ is measurable.}
\\
\text{ The map $(\mu,\sigma)$ is  continuous, and   Lipschitz with linear growth  }\\
\text{in its second argument, uniformly in the first one.} 
\end{array}
\ee

 This dynamics means the following. When no action is currently made on the system, i.e.~on the intervals in  $\Nc^{\phi}$, the system evolves according to a stochastic differential equation  driven by the Brownian motion $W$:
$$
dX_{s}=\mu\left(s,X_s\right)ds+\sigma\left(s,X_{s}\right)dW_{s} \;\;\mbox{ on $\Nc^{\phi}$}.
$$
When an order is send at $\tau^{\phi}_{i}$, we freeze the dynamics up to the end of the action (of the robot, of the execution/cancellation of the order) at time $\vartheta^{\phi}_{i}$. This amounts to saying that we do not observe the current evolution up to $\vartheta^{\phi}_{i}$, or equivalently that no corrective action will be taken before the end of the already launched operation at $\vartheta^{\phi}_{i}$. At the end of the action, the  state process takes a new value 
$$
X_{\vartheta^{\phi}_{i}}=F(\tau^{\phi}_{i}, X_{\tau^{\phi}_{i}-},  \alpha^{\phi}_{i}, \upsilon, \epsilon_{i}), \mbox{   $i\ge 1$}. 
$$ 
The fact that $F$ depends on the unknown parameter $\upsilon$ and the additional noise $\epsilon_{i}$ models the fact that the correct model  is not known with certainty, and that the exact value of the unknown parameter $\upsilon$ can (possibly) not be measured precisely just by observing  $(\vartheta^{\phi}_{i}-\tau^{\phi}_{i},X_{\vartheta^{\phi}_{i}}-X_{\tau^{\phi}_{i}-})$. 
 \vs2
 
 In order to simplify the notations, we shall now write:
  \begin{align}\label{eq : def z'}
{\rm z}':=(\varpi,F).
\end{align} 
% \begin{align}\label{eq: def Z Zcirc}
% Z^{z,\phi}:=(\cdot,X^{z,\phi})\;\;\mbox{ and }\;\;
% Z^{z,\circ}:=(\cdot,X^{z,\circ})
% \end{align}
% in which $X^{z,\circ}$ denotes the solution of \reff{eq: dyna X} for $\phi$ such that $\tau_{1}^{\phi}>T$ and satisfying $X_{t}^{z,\circ}=x$. This corresponds to the stochastic differential equation \reff{eq: dyna X} in the absence of impulse. Note in particular that 
% \begin{align}\label{eq: flow prop X}
% Z^{z,\phi}_{\vartheta^{\phi}_{1}}=\zr'(Z^{z,\circ}_{\tau^{\phi}_{1}-} ,\alpha^{\phi}_{1}, \upsilon, \epsilon_{1})\;\; \mbox{ on $\{\tau^{\phi}_{1}\ge t\}$,}
% \end{align} 
% in which 

 From now on, we denote by $\F^{z,m,\phi}=(\Fc^{z,m,\phi}_{s})_{t\le s\le 2T}$ the $\P_{m}$-augmentation of the filtration generated by 
$(X^{z,\phi},\sum_{i\ge 1}\1_{[\vartheta_{i}^{\phi},\infty)})
$  
on $[t,2T]$. We say that $\phi \in \Phi^{m}$ belongs to $\Phi^{z,m}$ if $(\tau_{i}^{\phi})_{i\ge 1}$ is a sequence of $\F^{z,m,\phi}$-stopping times  and $\alpha^{\phi}_{i}$ is $\Fc^{z,m,\phi}_{\tau^{\phi}_{i}}$-measurable, for each $i\ge 1$.  Hereafter an admissible control will be an element of $\Phi^{z,m}$.

%%%%%%%%%%%%%%%%%%%%%%
\subsection{Bayesian updates}

As already mentioned, acting on the system reveals some information on the true parameter value: the prior distribution evolves along time.   It should therefore be considered as a state variable to remain time consistent and be able to derive a dynamic programming equation. 
 Note also that its evolution can be of interest in itself. One can for instance be interested by the precision of our (updated) prior at the end of the control period, as it can serve as a new prior for another control problem.
 \vs2 
 
 In this section, we   describe   how it is updated with time, according to the usual Bayesian procedure. 
Given $z=(t,x)\in \Zb$, $u\in \U$ and $a\in \Ab$, we write the  law of ${\rm z}'[z,a,u,\epsilon_{1}]$, recall \reff{eq : def z'}, 
in the form 
	\[
		\qr(\cdot| z, a, u)d\Qr(\cdot|z, a),
	\]
in which $\qr(\cdot| \cdot)$ is a Borel measurable map and $\Qr(\cdot|z, a)$ is a dominating measure  on $\Zb$ for each $(z,a)\in \Zb\x \Ab$. This quantities can be inferred from the knowledge of $\zr'$ and the law of $\epsilon_{1}$, see Section \ref{sec: examples} for examples. 

Given initial conditions $z=(t,x)\in \Zb$, an initial prior $m\in \Mb$ and a trading strategy $\phi\in \Phi^{z,m}$, the conditional law  $M^{z,m,\phi}$ of $\upsilon$  at time $s\ge t$ is given by   
\be\label{eq: def M} 
M^{z,m,\phi}_{s}[C]:=\P_{m}[\upsilon \in C|\Fc^{z,m,\phi}_{s}],\;\;{ C\in \Bc(\U) }. 
\ee
As  no new information is revealed in between the end of an order and the start of the next one, the prior should remain constant on these time intervals: 
\begin{align}\label{eq: dyna m out of vartheta}
M^{z,m,\phi}=M^{z,m,\phi}_{\vartheta^{\phi}_{i}} \mbox{ on } [\vartheta^{\phi}_{i},\tau^{\phi}_{i+1}) \;,\;\;i\ge 0,
\end{align}
with the conventions $\vartheta^{\phi}_{0}=0$ and $M^{z,m,\phi}_{0}=m$. But, $M^{z,m,\phi}$  jumps at each time $\vartheta^{\phi}_{i}$ at which the effect of the last sent order is revealed, thus bringing a new information on the unknown parameter $\upsilon$. This prior update follows the classical  Bayes rule\footnote{In order to ensure that $M^{z,m,\phi}$ remains in $\Mb$ whenever $m\in \Mb$, we assume that 
$
{\mathfrak M}(\Mb;\cdot)\subset \Mb.
$}:
\begin{align} 
M^{z,m,\phi}_{\vartheta^{\phi}_{i}}  &={\mathfrak M}(M^{z,m,\phi}_{\tau^{\phi}_{i}-};Z^{z,\phi}_{\vartheta^{\phi}_{i}},Z^{z,\phi}_{\tau^{\phi}_{i}-},\alpha^{\phi}_{i}),\;\;i\ge 1,\label{eq: dyna m}
\end{align}
in which 
\be\label{eq: def M mathfrak}
{\mathfrak M}(m_{o};z'_{o},z_{o},a_{o})[C]:=\frac{\int_{C}  \qr(z'_{o}|z_{o},a_{o}, u) dm_{o}(u)}
{\int_{\U}  \qr(z'_{o}|z_{o},a_{o}, u) dm_{o}(u)},
\ee
for $(z_{o},z'_{o},a_{o},m_{o})\in \Zb^{2}\x \Ab\x \Mb$ and $C\in \Bc(\U)$. We refer to \cite{BBD16} for a formal proof of this intuitive fact.

\begin{Remark}\label{rem: measure abso cont} Again, the parameter is unknown, but we can have an idea of what values are more likely to be correct. As can be seen from \reff{eq: def M mathfrak}, $M^{z,m,\phi}$ remains absolutely continuous with respect to $m$ over time. The prior distribution should therefore have a support large enough  to include the true value, otherwise it will not be seen by the {\sl a posteriori} distributions as well. In practice, one can simply fix a uniform distribution on a rectangle, to which we are certain that the true parameter  belongs. If one is only interested by a crude approximation, but want to minimize the computation time, one can simply specify several low/medium/high values and concentrate the support of  $m_{0}$  on these (i.e.~start with a combination of Dirac masses). 
\end{Remark}

 \begin{Remark}\label{rem: joint condi distri ZM}
For later use, note  that the above provides  the joint conditional distribution of  $(Z^{z,\phi}_{\vartheta^{\phi}_{i}},M^{z,m,\phi}_{\vartheta^{\phi}_{i}})$ given ${\Fc^{z,m,\phi}_{\tau_{i}}}$:
 \begin{equation}\label{eq: joint cond law X,M}
\P[(Z^{z,\phi}_{\vartheta^{\phi}_{i}},M^{z,m,\phi}_{\vartheta^{\phi}_{i}})\in B\x D |{\Fc^{z,m,\phi}_{\tau^{\phi}_{i}-}}]= \kr(B\x D|Z^{z,\phi}_{\tau^{\phi}_{i}-}, M^{z,m\phi}_{\tau^{\phi}_{i}-},\alpha^{\phi}_{i})
 \end{equation}
 in which
 \b*
 \kr(B\x D| z_{o},m_{o},a_{o}):=
\int_{\U} \int_{B}  \1_{D}({\mathfrak M}(m_{o};z',z_{o},a_{o})) \qr(z'|z_{o},a_{o}, u)d\Qr(z'|z, a) dm_{o}(u).
 \e*
 
\end{Remark}

%%%%%%%%%%%%%%%%%%%%%%
	\subsection{Gain function}

Given $z=(t,x) \in \Zb$ and $m\in \Mb$, the aim of the controller is to maximize the expected value of the gain functional\footnote{ $g$ is assumed to be measurable and (for simplicity)  bounded on $\Zb \times \Mb \times \U\times \Ec$. }
    \[
        \phi \in \Phi^{z,m} \mapsto {G}^{z,m}(\phi) := g(Z^{z,\phi}_{\T[\phi]}, M^{z,m,\phi}_{\T[\phi]}, \upsilon, \epsilon_{0}),
    \]
 in which $\T[\phi]$ is the end of the last action after $T$:
 $$
 \T[\phi]:=\sup\{\vartheta^{\phi}_{i}: i\ge 1,\; \tau^{\phi}_{i}\le T\}\vee T.
 $$   
 Note that we do not look at the value of $Z^{z,\phi}$ at $T$ but rather at $\T[\phi]$ which is either $T$ or the end of the last trade {sent} before $T$. This is motivated by the use of robots: we do not want to stop it at $T$ if it is running, we rather prefer to wait till the end of the algorithm. This is compensated by the fact that a penalty can be imposed when $T[\phi]$ is strictly larger that $T$, through the objective function $g$.  Also note that the terminal reward depends on the parameter $\upsilon$. This is motivated by applications to optimal liquidation in which a final {\sl {large}} order may be {sent} at the end, to liquidate {\sl immediately} the remaining {shares}. 

 As suggested earlier,  the gain may not only depend on the value of the original time-space state process $Z^{z,\phi}_{\T[\phi]}$ but also on $M^{z,m,\phi}_{\T[\phi]}$, to model the fact that we are also interested by the precision of the estimation made on $\upsilon$ at the final time.   Also note that one could add a running cost without additional difficulty, it can actually be incorporated into the state process $Z^{z,\phi}$.
 \vs2

Given $\phi \in \Phi^{z,m}$, the expected reward is 
    \[
        J(z, m; \phi) := \mathbb{E}_{m}\left[ G^{z,m}(\phi)\right],
    \]
and 
    \begin{align}\label{eq: def vr}
        \vr(z,m) := \sup_{\phi \in \Phi^{z,m}}J(z, m; \phi)\1_{\{t\le T\}}+\1_{\{t> T\}}\mathbb{E}_{m}\left[ g(z, m, \upsilon, \epsilon_{0})\right] 
    \end{align}
   is the corresponding value function. Note that $\vr$ depends on $m$ through the set of admissible controls $ \Phi^{z,m}$ and the expectation operator $ \mathbb{E}_{m}$, even if $g$ does not depend on $M^{z,m,\phi}_{\T[\phi]}$. 
 %%%%%%%%%%%%%%%%%%%%%%%%%%%%%%%%%%%%%%%%%%%%%%%%%
  \section{Value function characterization and numerical approximation}

\subsection{The dynamic programming quasi-variational equation}
The aim of this section is to explain how one can derive a pde  characterization of the optimal expected gain. 
As usual, it should be related to a  dynamic programming principle. In our setting, it should read as follows: Given $z=(t,x) \in \Zb$ and $m\in \Mb$, then 
	\begin{equation}\label{eq: DPP formal}
		\vr(z,m)= \sup_{\phi \in \Phi^{z,m}}\E_{m} [\vr(Z^{z,\phi}_{\theta^{\phi}},M^{z,m,\phi}_{\theta^{\phi}})],
	\end{equation}
for all collection $(\theta^{\phi},\phi\in \Phi^{z,m})$  of $\F^{z,m,\phi}$-stopping times with values in $[t, 2T]$ such that 
$$
\theta^{\phi}\in \Nc^{\phi}\cap [t,\T[\phi]]~\P_{m}-\as \;,
$$
recall the definition of $\Nc^{\phi}$ in \reff{eq: def Nphi}. 

Let us comment this. First, one should restrict to stopping times such that $\theta^{\phi}\in \Nc^{\phi}$. The reason is that no new impulse can be made outside of $\Nc^{\phi}$, each interval $[\tau_{i}^{\phi},\vartheta^{\phi}_{i})$ is a latency period. Second, the terminal gain is evaluated at $\T[\phi]$, which in general is different from $T$. Hence, the fact that $\theta^{\phi}$ is only bounded by $\T[\phi]$.

\vs2
We continue our discussion, assuming that \reff{eq: DPP formal} holds and that $\vr$ is sufficiently smooth.
Let us denote $Z^{z,\circ}$ the dynamics of the state process when no order is sent. Then, the above implies in particular
$$
\vr(z,m)\geq  \E_{m} [\vr(Z^{z,\circ}_{t+h},m)]
$$ 
for $0<h\le T-t$. This corresponds to the sub-optimality   of the control consisting in making no impulse on  $[t,t+h]$. Applying It\^{o}'s lemma, dividing by $h$ and letting $h$ go to $0$, we obtain 
$$
-\Lc\vr(z,m)\ge 0 
$$
in which  
 $$
 \Lc\vp:=\partial_{t}\vp  +\langle \mu, D\vp\rangle + \frac12 {\rm Tr}[\sigma\sigma^{\top} D^{2}\vp]. 
 $$
 On the other hand, it follows from \reff{eq: DPP formal} and Remark \ref{rem: joint condi distri ZM} that 
\begin{align*}
 \vr(z,m)&\geq  \sup_{a\in \Ab}\E_{m} [\vr({\rm z}'[z,a,\upsilon,\epsilon_{1}],{\mathfrak M}(m;{\rm z}'[z,a,\upsilon,\epsilon_{1}],z,a)) ] = \Kc\vr(z,m)
\end{align*}
where 
\begin{align}\label{eq : def Kc}
\Kc \vp:=\sup_{a\in \Ab}\int \vp(z',m')d\kr(z',m'|\cdot,a).
 \end{align}
 This corresponds to the sub-optimality of sending an order immediately. 
 As for the time-$T$ boundary condition, the same reasoning as above implies 
 $$
 \vr(T,\cdot)\ge \Kc_{T}g\;\;\mbox{ and } \;\;\vr(T,\cdot)\ge \Kc\vr(T,\cdot), 
 $$
 in which 
\begin{align}\label{eq : def KcT}
 \Kc_{T}g(\cdot,m)=\int_{\U}\int_{\Ec} g(\cdot,m,u,e)d\P_{\epsilon}(e)dm(u).
\end{align} 
 
 By optimality, $\vr$ should therefore solve the quasi-variational equations
 \begin{align}
 \min\left\{-\Lc\vp\;,\;\vp-\Kc\vp\right\}=0 &\;\mbox{ on } [0,T)\x \R^{d}\x\Mb\label{eq: pde interior}\\
  \min\left\{\vp- \Kc_{T}g,\vp- \Kc \vp\right\}=0 &\;\mbox{ on } \{T\}\x \R^{d}\x\Mb.\label{eq: pde T} 
  \end{align}

To ensure that the above operator is continuous, we assume   that, on  $\R_{+}\x \R^{d}\x \Mb$, 
\be\label{eq: hyp conti KcT et Kc}
\begin{array}{c}
\text{$\Kc_{T}g$ is continuous, and $\Kc\vp$ is upper- (resp.~lower-) semicontinuous,}\\
\text{for all   upper- (resp.~lower-) semicontinuous bounded function $\vp$.}
\end{array}
\ee

Finally, we assume that comparison holds for \reff{eq: pde interior}-\reff{eq: pde T}. A sufficient condition is provided in  \cite{BBD16}.

\begin{Assumption}\label{ass: comp}
Let $U$ (resp.~$V$) be a upper- (resp.~lower-) semicontinuous bounded viscosity sub- (resp.~super-) solution of  \reff{eq: pde interior}-\reff{eq: pde T}. Assume further that $U \le V$ on $(T,\infty)\x \R^{d}\x \Mb$.
Then, $U \le V$ on $\Zb\x \Mb$.
\end{Assumption}

We can now state the main result of \cite{BBD16}.   

\begin{Theorem}[\cite{BBD16}]\label{thm:viscosity}
Let Assumption \ref{ass: comp} hold. Then,  $\vr$ is continuous on $\Zb\x \Mb$ and is the unique bounded viscosity solution of \reff{eq: pde interior}-\reff{eq: pde T}.
\end{Theorem}

 \subsection{An example of numerical scheme}\label{subsec: numerical resolution}

When the comparison result of Assumption \ref{ass: comp} holds, one can easily derive a convergent finite different scheme for \reff{eq: pde interior}-\reff{eq: pde T}.

\vs2

We  consider here a simple explicit scheme based on \cite{camilli1995approximation, camilli2009finite}. We let $h_{0}$ be a time-discretization step so that $T/h_{0}$ is an integer, and set $\mathbf{T}^{h_{0}} := \{t_{j}^{h_{0}} := jh_{0}, j \leq T/h_{0}\}$.  The space $\mathbb{R}^{d}$ is discretized with a space step $h_{1}$ on a rectangle $[-c,c]^{d}$, containing $N_{h_{1}}^{x}$ points on each direction.   The corresponding finite set is denoted by $\Xb^{h_{1}}_{c}$. 

The first order derivatives  $\partial_{t} \vp$ and $(\partial\vp/\partial x^{i})_{i\le d}$ are approximated by using the standard up-wind approximations:
\begin{align*}
 \Delta^{h_{0}}_{t}\vp(t,x,m)&:=h_{0}^{-1}(\vp(t+h_{0},x,m)-\vp(t,x,m))\\
 \Delta^{h_{0}}_{h_{1},i}\vp(t,x,m)
&:=\left\{ \begin{array}{lcl}
h_{1}^{-1}(\vp(t+h_{0},x+e_{i}h_{1},m)-\vp(t,x,m)) & \mbox{if} & \mu^{i}(x)\ge 0\\
h_{1}^{-1}(\vp(t,x,m)-\vp(t+h_{0},x-e_{i}h_{1},m)) & \mbox{if} & \mu^{i}(x)< 0,
\end{array} \right.
\end{align*}
in which $e_{i}$ is $i$-th unit vector of $\R^{d}$. 
  
As for the second order term, we use the fact that each point $ x\in \R^{d}$ can be approximated as a weighted combination
$$
  x=\sum_{  x'\in C_{h_{1}}(  x)} x' \omega(  x'|  x)
$$ 
of the points $  x'$ lying on the corners $C_{h_{1}}(   x)$ of the cube formed by the partition of $\R^{d}$ it belongs too. Then, given another small parameter $h_{2}>0$, we approximate
${\rm Tr}[\sigma(x)\sigma(x)^{\top} D^{2}\vp(t,x,m)]$ by ${\rm T}_{h_{0}, h_{1}}^{h_{2}}[\vp](t,x,m)$ defined as
\begin{align*}
(h_{2}d)^{-1} \sum_{i=1}^{d}& [\vp]_{h_{1}}(t+h_{0},x+\sqrt{h_{2}} \sigma^{  i}(x),m)+[\vp]_{h_{1}}(t+h_{0},x-\sqrt{h_{2}}\sigma^{  i}(x),m)\\
&   -2h_{2}^{-1} \vp(t,x,m)
\end{align*}
in which $\sigma^{i}$ is the $i$-th column of $\sigma$ and 
$$
 [\vp]_{h}(t,x,m):=  \sum_{  x'\in C_{h_{1}}(  x)}  \omega(  x'|  x) \vp([t]_{h},x',m) \mbox{ with } [t]_{h}:=\min [t,2T]\cap \left(\mathbf{T}^{h_{0}}\cup [T,2T] \right), 
$$
is a piecewise linear approximation of $\vp$. In the case where only the first row $\sigma^{ 1 \cdot}$ of $\sigma$ is not identically equal to $0$, one can use the usual simpler approximation 
\begin{align*}
&(h_{1})^{-1} \|\sigma^{1\cdot}\|^{2}\left( \vp(t+h_{0},x+\sqrt{h_{1}} e_{1},m)+\vp(t+h_{0},x-\sqrt{h_{1}}e_{1},m)\right)\\
&-2(h_{1})^{-1}\|\sigma^{1\cdot}\|^{2}  \vp(t,x,m).
\end{align*}

Similarly, we  approximate $\Kc \vp$ by 
\begin{align*}
  \Kc_{h_{0},h_{1}}\vp(t,x,m)&:=\sup_{a\in \Ab} \int [\vp]_{h}(\max(t+ h_{0}, t'),x',m')d\kr( t', x', m'|t,x,m,a).
\end{align*}
Letting $h:=(h_{0},h_{1},h_{2})$, and setting
\begin{align}\label{scheme_L}
\Lc^{h}\vp&:= \Delta^{h_{0}}_{t}\vp+\sum_{i\le d} \mu^{i}  \Delta^{h_{0}}_{h_{1},i}\vp+ \frac12 {\rm T}_{h_{0}, h_{1}}^{h_{2}}[\vp],
\end{align}
our numerical scheme consists in solving
 \begin{align}
 \min\left\{-\Lc^{h}\vp\;,\;\vp-\Kc_{h_{1}}\vp\right\}=0 &\;\mbox{ on } (\mathbf{T}^{h_{0}}\setminus\{T\})\x (\Xb^{h_{1}}_{c}\setminus\partial \Xb^{h_{1}}_{c})\x\Mb,\label{eq : num pde}\\
\min\{\vp-\Kc_{T}g\;,\;\vp-\Kc_{h_{1}} \vp\}=0 &\;\mbox{ on } \{T\}\x (\Xb^{h_{1}}_{c}\setminus\partial \Xb^{h_{1}}_{c})\x\Mb,\label{eq: num pde T}
\\
\vp-\Kc_{T}g:=0 &\;\mbox{ on } ([0,T]\x \partial \Xb^{h_{1}}_{c}\x\Mb)\cup ((T,2T]\x \R^{d}\x \Mb) .\label{eq: num pde bord X}
 \end{align}
 We specify here a precise boundary condition on $ \partial \Xb^{h_{1}}_{c}$ but any other (bounded) boundary condition could be used. Finally, we extend $\vr^{c}_{h}$ to the whole space by setting $\vr^{c}_{h}=[\vr^{c}_{h}]_{h}$ on  $[0,T]\x \R^{d}\x\Mb$. 
 
This scheme is always convergent as $(h_{2},h_{1}/h_{2},h_{0}/h_{1})\to 0$ and $c\to \infty$.

\begin{Proposition} Let $\vr^{c}_{h}$ denote the solution of \reff{eq : num pde}-\reff{eq: num pde T}-\reff{eq: num pde bord X}. If Assumptions \ref{ass: comp} holds, then $\vr^{c}_{h}\to \vr$ as $(h_{2},h_{1}/h_{2},h_{0}/h_{1})\to 0$ and then $c\to \infty$. 
\end{Proposition}
\proof	Using Lemma \ref{ass : passage limite ope non local discretisation} below, one easily checks that our scheme satisfies the conditions of   \cite[Theorem 2.1.]{BaSo91}. In particular,   $|\vr^{c}_{h}|\le \sup |g|<\infty $. Then, the convergence holds by the same arguments as in  \cite[Theorem 2.1.]{BaSo91}, it suffices to replace their assertion (2.7)  by Lemma \ref{lem: approx stabi} stated below. \ep

\begin{Remark} We did not discuss in the above the problem of the discrete approximation of $\Mb$. Applications will typically be based on a parameterized family $\Mb=\{m_{\theta},\theta \in \Theta\}$, for a subset $\Theta$ of a finite dimensional space.  We can then further approximate $\Theta$ by a sequence of finite sets to build up a numerical scheme.   Similarly, the set of control values $\Ab$ need to be approximated in practice. If the corresponding sequences of approximations are dense, then convergence of the numerical scheme will still hold. \end{Remark}

We conclude this section with the technical lemmas that were used in the above proof. 
 
 \begin{Lemma}\label{ass : passage limite ope non local discretisation} If  $(u_{n})_{n\ge 1}$ is a   bounded sequence of functions on  $\Zb\x \Mb$ and $(z_{n},m_{n})_{n\ge 1}$ is a sequence in $\Zb\x \Mb$ that converges to $(z_{\circ},m_{\circ})$, then 
$$
\liminf_{\tiny \begin{array}{c}n\to \infty\\(h_{0}, h_{1})\to (0,0)\end{array}}  \Kc_{h_{0}, h_{1}} u_{n}(z_{n},m_{n})\ge \Kc u_{\circ}(z_{\circ},m_{\circ})
\;\;\;,\mbox{where } u_{\circ}:=\liminf \limits_{\tiny \begin{array}{c}n\to \infty \\ (z',m')\to \cdot \end{array} } u_{n}(z',m'), 
$$
and
$$
\limsup_{\tiny \begin{array}{c}n\to \infty\\(h_{0}, h_{1})\to (0,0)\end{array}}  \Kc_{h_{0}, h_{1}} u_{n}(z_{n},m_{n})\le \Kc u^{\circ}(z_{\circ},m_{\circ})
\;\;\;,\mbox{where } u^{\circ}:=\limsup \limits_{\tiny \begin{array}{c}n\to \infty \\ (z',m')\to \cdot \end{array} } u_{n}(z',m').
$$
\end{Lemma}
 
\begin{proof}
We first rewrite
\begin{equation}\label{Kh}
	\Kc_{h_{0}, h_{1}} u_{n}(z_{n},m_{n}) = \sup_{a \in \mathbf{A}}\int u_{n,h}(z', m')dk(z', m' | z_{n}, m_{n},a)
\end{equation}
where $u_{n,h}(z', m') := [u_{n}]_{h_{1}}(\max(t_{n} + h_{0}, t'), x', m')$. Let  $\bar{u}_{n_{\circ}, h_{\circ}}$ be the lower-semicontinuous envelope of $\inf_{n \geq n_{\circ}, h \leq h_{\circ}}u_{n,h}$. From (\ref{Kh}), we get, for $n \geq n_{\circ}$ and $h \leq h_{\circ}$,
	\[
		\mathcal{K}u_{n,h}(z_{n},m_{n}) \geq \mathcal{K}\bar{u}_{n_{\circ}, h_{\circ}}(z_{n}, m_{n}),
	\]
and, by (\ref{eq: hyp conti KcT et Kc}), passing to the limit inf as $(n, h) \rightarrow (+\infty, 0)$ leads to
	\[
		\liminf_{(n, h) \rightarrow (+\infty, 0)} \mathcal{K}u_{n,h}(z_{n},m_{n}) \geq \mathcal{K}\bar{u}_{n_{\circ}, h_{\circ}}(z_{\circ}, m_{\circ}).
	\]
Moreover, $\bar{u}_{n_{\circ}, h_{\circ}}$ $\uparrow$ $u_{\circ}$ point-wise. The required result is then obtained by monotone convergence.
\end{proof}

 \begin{Lemma}\label{lem: approx stabi}  Let   $(u_{n})_{n\ge 1}$ be a sequence of  lower semi-continuous maps on  $\Zb\x \Mb$ and define 
 $u_{\circ}:=\liminf_{(z',m',n)\to (\cdot,\infty)} u_{n}(z',m')$ on $ \Zb\x \Mb$. Assume that $u_{\circ}$ is locally bounded.  Let $\vp$ be a  continuous map and assume that $(z_{\circ},m_{\circ})$ is a strict minimal point of $u_{\circ}-\vp$ on $\Zb\x \Mb$. Then, one can find   a bounded open set  $B$  of $[0,T]\x \R^{d}$ and  a sequence $(z_{k},m_{k},n_{k})_{n\ge 1}\subset B\x\Mb\x\N$ such that $n_{k}\to \infty$, $(z_{k},m_{k})$ is a   minimum point of  $u_{n_{k}}-\vp$ on $ B \x \Mb$ and   $(z_{k},m_{k},u_{n_{k}}(z_{k},m_{k}))\to (z_{o},m_{\circ},u_{\circ}(z_{o},m_{o}))$.
 \end{Lemma} 

\proof Since $\Mb$ is assumed to be locally compact, it suffices  to repeat the arguments in the proof of \cite[p80, Proof of Lemma 6.1]{barles1994solutions}.
\ep

%%%%%%%%%%%%
\subsection{Construction of $\eps$-optimal controls}\label{subsec: eps optimal controls}
 
 It remains to explain how to deduce the optimal policy.   At each of point $(t,x)$ of the time-space grid and for each prior $m$, one computes  
$$
(\hat \ell(t,x,m),\hat b(t,x,m))\in {\rm arg} \max \left\{\int \vr^{c}_{h}(z', m')d\kr(z', m' | (t,x), m,(\ell,b)), \;(\ell,b)\in A\right\}.
$$
If  $\vr^{c}_{h}(t,x,m)$ is equal to the above maximum, then we play the control $(\hat \ell(t,x,m),\hat b(t,x,m))$, otherwise we wait for the next time step. This is the usual philosophy: we act on the system only if this increases the expected gain. As already argued, here the gain should not only be considered as an improvement of the current future reward, it can also be a gain in the precision of our prior which will then lead to better {future} rewards. 
\vs2

This produces a Markovian control which is optimal for the discrete time problem associated to our numerical scheme, and  asymptotically optimal for the original control problem. 
We shall use this algorithm for the toy examples presented in the next section. 
%%%%%%%%%%%%%%%%%%%%%%%%%%%%%%%%%%%%%%%%%%%%%%%%%%%%%%

\section{Applications to optimal trading}\label{sec: examples}

This section is devoted to the study of two examples of application. Each of them corresponds to an idealized model, the aim here is not to come up with a good model but rather to show the flexibility of our approach, and to illustrate numerically the behavior of our backward algorithm.

%%%%%%%%%%%%%%
    \subsection{Immediate impact of aggressive orders}
\def\Bb{\mathbf B}

We consider first a model in which the impact of each single order sent to the market is taken into account. It means that $\alpha_{i}$ represents the number of shares bought exactly at time $\tau_{i}$, so that $\ell_{i} = 0$, for each $i$. This corresponds to $\Ab = \{0\} \times \Bb$ in which $\Bb\subset \R_{+}$ is a compact  set of values of admissible orders. Therefore,  one can identify $\Ab$ to $\Bb$ in the following, and we will only write $b$ for $a=(0,b)\in \Ab$ and $\beta_{i}$ for $\alpha_{i}=(\ell_{i},\beta_{i})$.  

Our model  can be viewed as a scheduling model or as a model for illiquid market.
The first component of $X$ represents the stock price. We consider a simple linear impact: when a trade of size $\beta_{i}$ occurs at $\tau_{i}$, the stock price jumps by
    \[
        X_{\vartheta_{i}}^{1} = X_{\tau_{i}-}^{1} + \beta_{i}(\upsilon + \epsilon_{i})/2
    \]
in which $\upsilon \in \mathbb{R}$ is the unknown linear impact parameter, $(\epsilon_{i})_{i\ge 1}$ is a sequence of independent noises following a centered Gaussian distribution with standard deviation $\sigma_{\epsilon}$. The coefficient 1/2 in the dynamics of $X^{1}$ stands for a 50\% proportion of immediate resilience.

It evolves according to a Brownian diffusion between two trades and has a residual resilience effect:
\begin{equation}
dX_{t}^{1} = \sigma dW_{t}^{1} + dX_{t}^{4} \text{ and } dX_{t}^{4} = -\rho X_{t}^{4}dt,% t \not\in \{\tau_{i}, i \geq 1\},
\end{equation}
where $\sigma, \rho > 0$ and $X^{1}_{0} \in \mathbb{R}$ are constants. The process $X^{4}$ represents the drift of $X^{1}$ due to the non immediate resilience and $X_{0}^{4} = 0$. When a trade occurs, it jumps according to

    \[
        X_{\vartheta_{i}}^{4} = X_{\tau_{i}-}^{4} + \beta_{i}(\upsilon + \epsilon_{i})/2.
    \]
    We call it spread hereafter. This is part of the deviation from the un-impacted dynamic. 

The third component, which describes the total cost, evolves as

    \[
        X_{\vartheta_{i}}^{2} = X_{\tau_{i}-}^{2} + X_{\tau_{i}-}^{1}\beta_{i} + (\upsilon + \epsilon_{i})\frac{\beta_{i}^{2}}{2}.
    \]

Finally, the last component is used to keep track of the cumulative number of shares bought:

    \[
        X_{\vartheta_{i}}^{3} = X_{\tau_{i}-}^{3} + \beta_{i}.
    \]

We are  interest in the cost of buying $N$ shares, and maximize the criteria  
$$
-\E_{m}[e^{\eta L(X_{T}, \upsilon)}\wedge  C ]
$$
where $\eta>0$ is a risk aversion parameter, $C>0$,  and 
    \[
        L(X_{T}, \upsilon) := X_{T}^{2} + X_{T}^{1}(N - X_{T}^{3}) + (\upsilon + \epsilon_{0})\frac{(N - X_{T}^{3})^{2}}{2}
    \]
represents the total cost after setting the total number of shares bought to $N$ at $T$.

If the prior law $m$ on $\upsilon$ is a Gaussian distribution, then $\qr(\cdot| t,x, b, u)$ is a Gaussian density with respect to
    \[
        d\Qr(x'|t,x, b) = dx^{1'}d\delta_{x^2 + b x^{1'}}(x^{2'})d\delta_{x^3+b}(x^{3'})d\delta_{x^4 + (x^{1'} - x^1)}(x^{4'})
    \]
and the transition map
    \[
        \mathfrak{M}(m;t',x', t, x,b)[C] = \frac{\int_{C}\qr(x'| t, x, b, u)dm(u)}{\int_{\mathbb{R}}\qr(x'| t, x,b, u)dm(u)},
    \]
maps Gaussian distributions into Gaussian distributions, which, in practice, enables us to restrict $\Mb$ to the set of Gaussian distributions.
More precisely, if $(m_{\upsilon}(\tau_{i}-), \sigma_{\upsilon}(\tau_{i}-))$ are the mean and the standard deviation of $M_{\tau_{i}-}$, then the values corresponding to the posterior distribution $M_{\vartheta_{i}}$ are  
    \[
        \begin{aligned}
            \sigma_{\upsilon} (\vartheta_{i}) &= \mathbf{1}_{\{\sigma_{\upsilon}(\tau_{i}-)\not=0\}}\left(\frac{1}{\sigma_{\upsilon}(\tau_{i}-)^{2}} + \frac{1}{\sigma_{\epsilon}^{2}}\right)^{-\frac12}, \\
            m_{\upsilon}(\vartheta_{i}) &= m_{\upsilon}(\tau_{i}-)\mathbf{1}_{\{\sigma_{\upsilon}(\tau_{i}-)=0\}} + \left( \frac{X_{\vartheta_{i}}^{1} - X_{\tau_{i}-}^{1}}{\sigma_{\epsilon}^{2}} + \frac{m_{\upsilon}(\tau_{i}-)}{\sigma_{\upsilon}(\tau_{i}-)^{2}}\right)\mathbf{1}_{\{\sigma_{\upsilon}(\tau_{i}-)\not=0\}}.
        \end{aligned}
    \]
\def\wr{{\rm w}}

Comparing to the general result of the previous section, we  add a boundary condition $\vr(t,x^{1},x^{2}, N,x^{4})=1$ and restrict the domain of $X^{3}$ to be $\{0,\ldots,N\}$. Since this parameter $x^{3}$ is discrete this does not change the nature of our general results. 

Note also that the map $\Psi(t,x,m)=N-x^{3}$ defined on $[0,T]\x\R^{2}\x\{0,\ldots,N\}\x \R\x \Mb$ actually satisfies the conditions provided in \cite{BBD16} to ensure that Assumption \ref{ass: comp} holds.

\vs2

We now discuss a numerical illustration.  We consider 30 seconds of trading and $N = 25$ shares to buy. We take $\eta=1$, $x_{0} = 100$ and $\sigma = 0.4 x_{0}$ which corresponds to a volatility of $40\%$ in annual terms. The trading period is divided into intervals of 1 second-length. The size of an order $\beta_{i}$ ranges in $\{ 1, 2, 3, 4, 5\}$. We take $\sigma_{\varepsilon} = 10^{-4}$ and $\rho$ such that the spread $X^{4}$ is divided by 3 every second if no new order is sent. We start with a prior given by a Gaussian distribution with mean $m_{\upsilon}(0)$ and standard deviation $\sigma_{v}(0)$. Finally, we take $C=10^{200}$ which makes this threshold parameter essentially inefficient while still ensuring that the terminal condition is bounded.

In Figure \ref{resilience3D}, we plot the optimal strategy for $\sigma_{\upsilon}(0) = 5.10^{-4}$ and $m_{v}(0) = 5.10^{-2}$ in terms of $(X^{2}, X^{3})$. Clearly, the level of spread $X^{4}$ has a significant impact: when it is large, it is better to wait for it to decrease before sending a new order. This can also be observed in Figure \ref{resilience2D} which provides a simulated path corresponding to an initial prior $(m_{v}(0) = 2.10^{-2}, \sigma_{\upsilon}(0) = 10^{-3})$: after 15 seconds the algorithm alternates between sending an order and doing nothing, i.e. waiting for the spread to be reduced at the next time step. On the top right graph, we can also observe that the low mean of the initial prior combined with a zero initial resilience leads to sending an order of size 3 at first, then the mean of the prior is quickly adjusted to a higher level and the algorithm slows down immediately.

  \begin{figure}
 \begin{center}
 \includegraphics[scale=0.31]{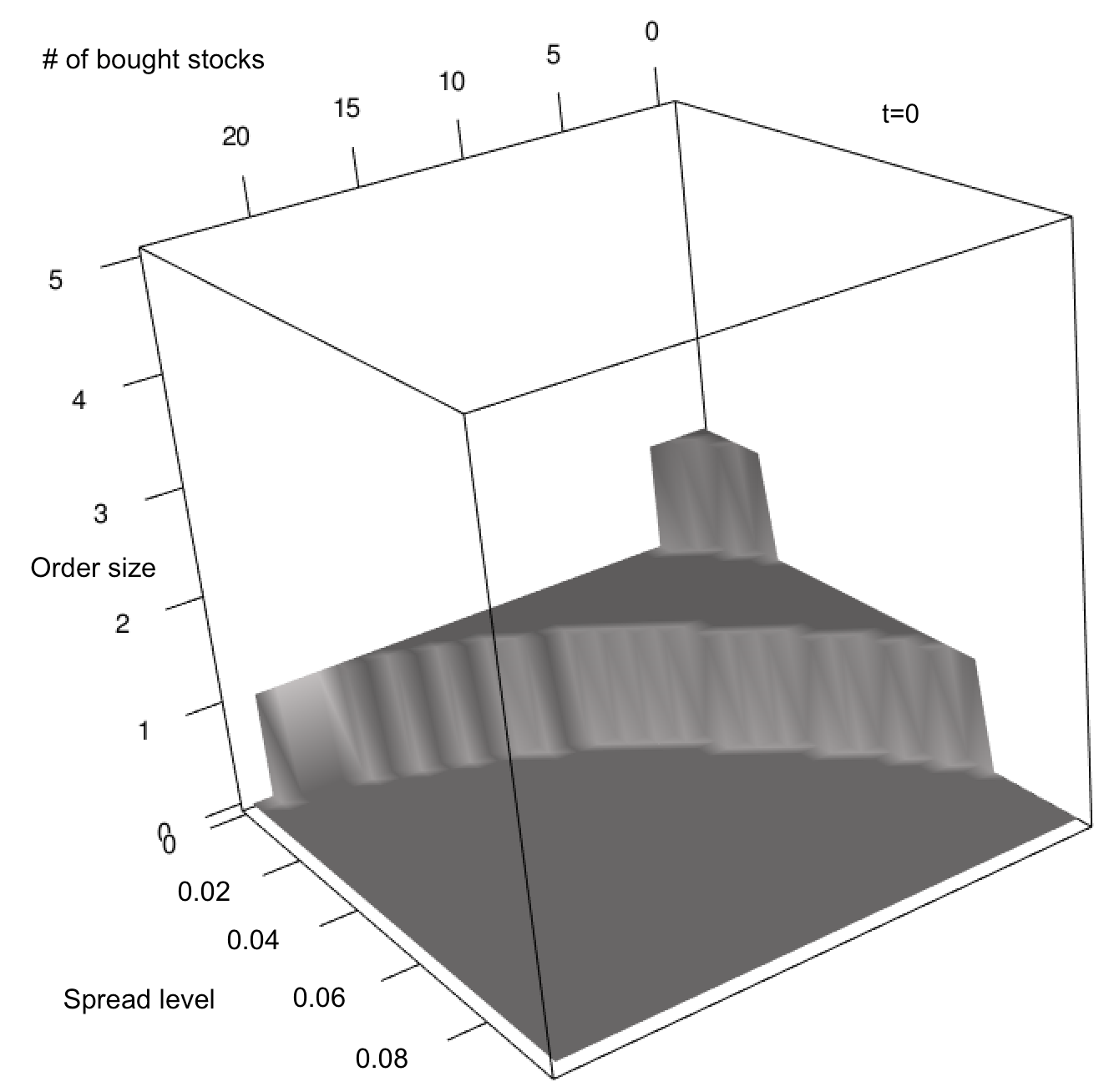}
 \end{center}
 \includegraphics[scale=0.31]{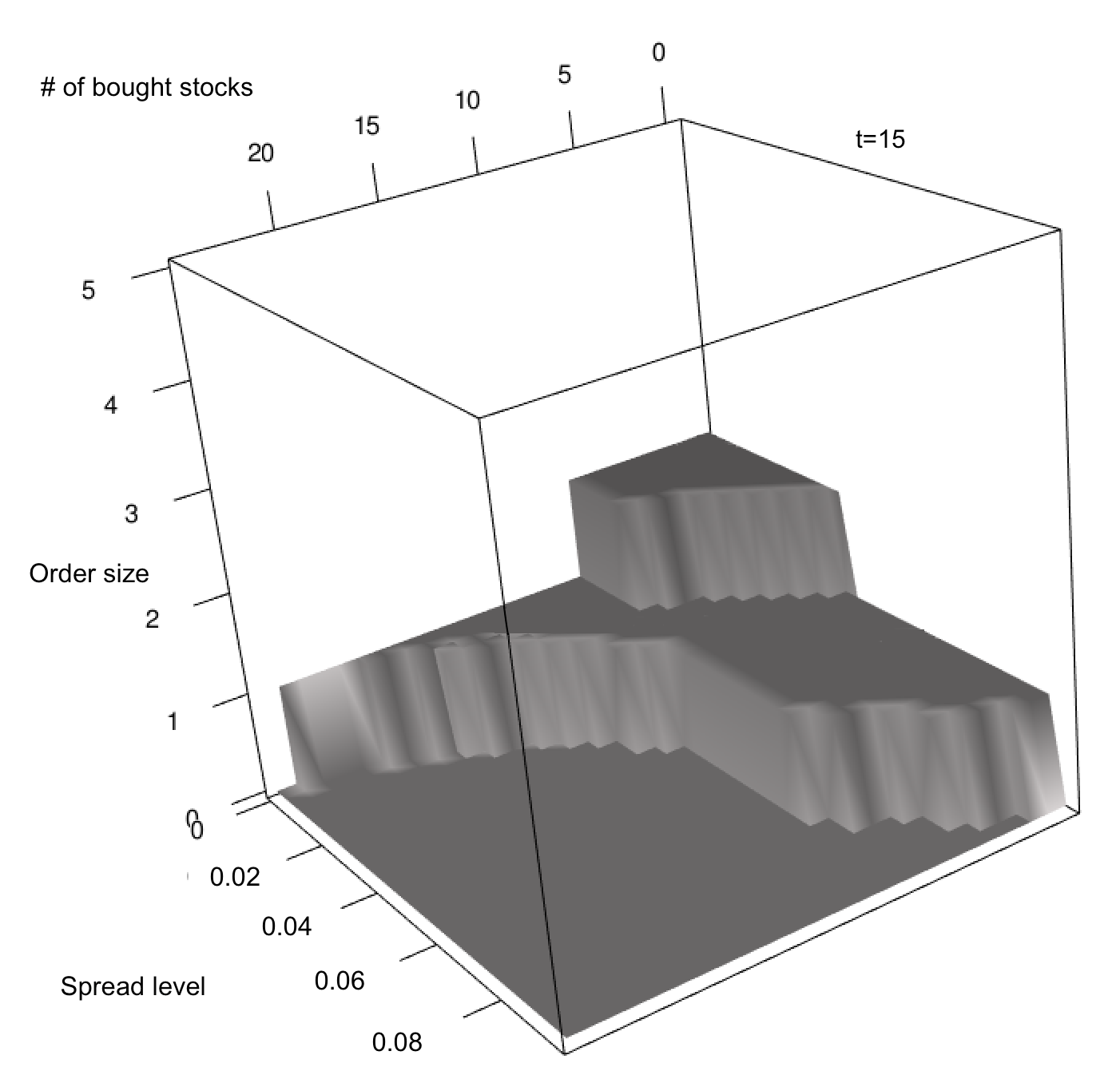} \includegraphics[scale=0.31]{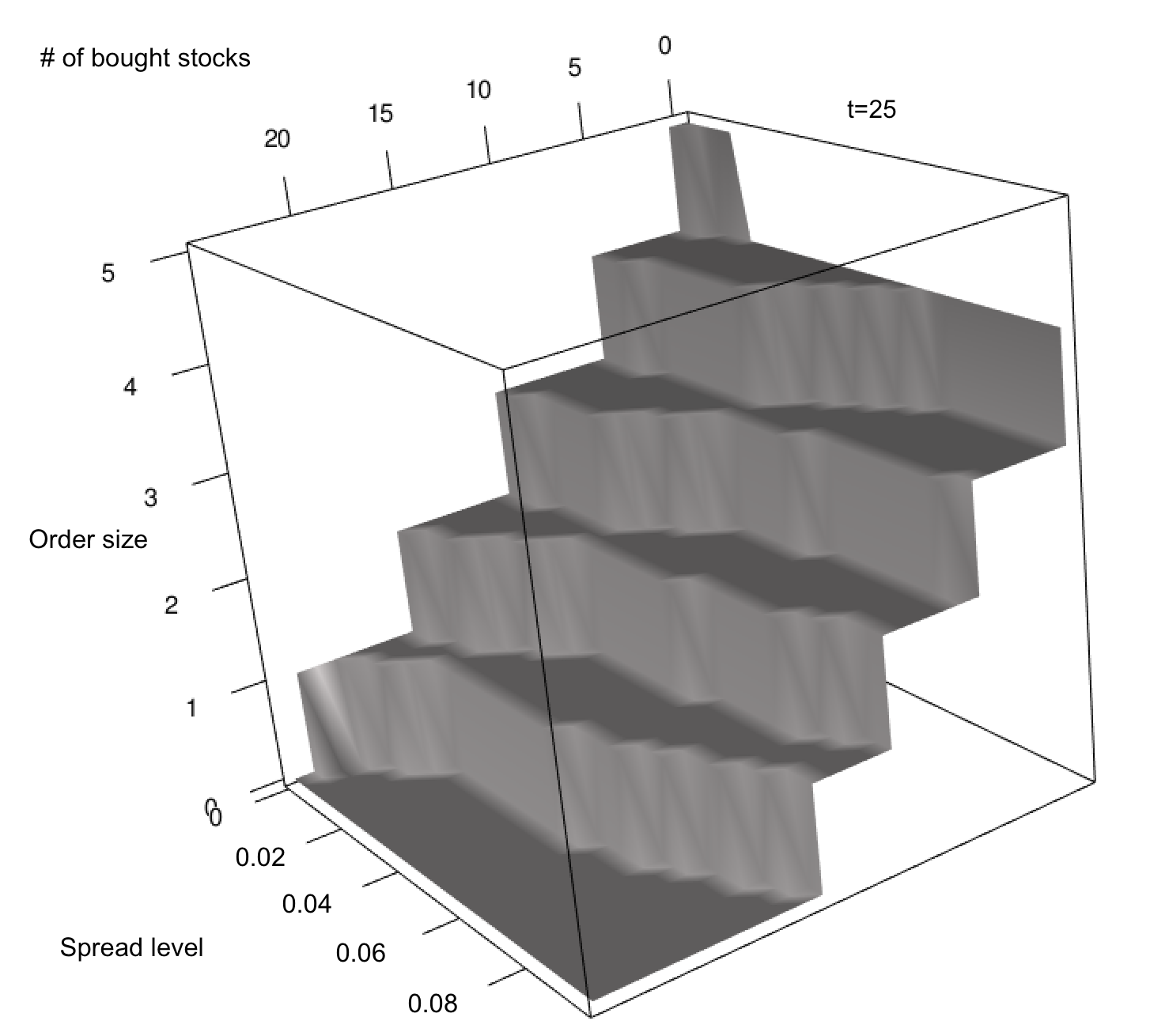}
 \captionof{figure}{Evolution of $\beta$ in terms of $(X^{3},X^{4})$ at time $0$s (top), $15$s (left) and $25$s (right), for $(m_{\upsilon},\sigma_{\upsilon})=(5.10^{-2},5.10^{{-4}})$.}
\label{resilience3D}
\end{figure}

 \begin{figure}
 \begin{center}
 \includegraphics[scale=0.29]{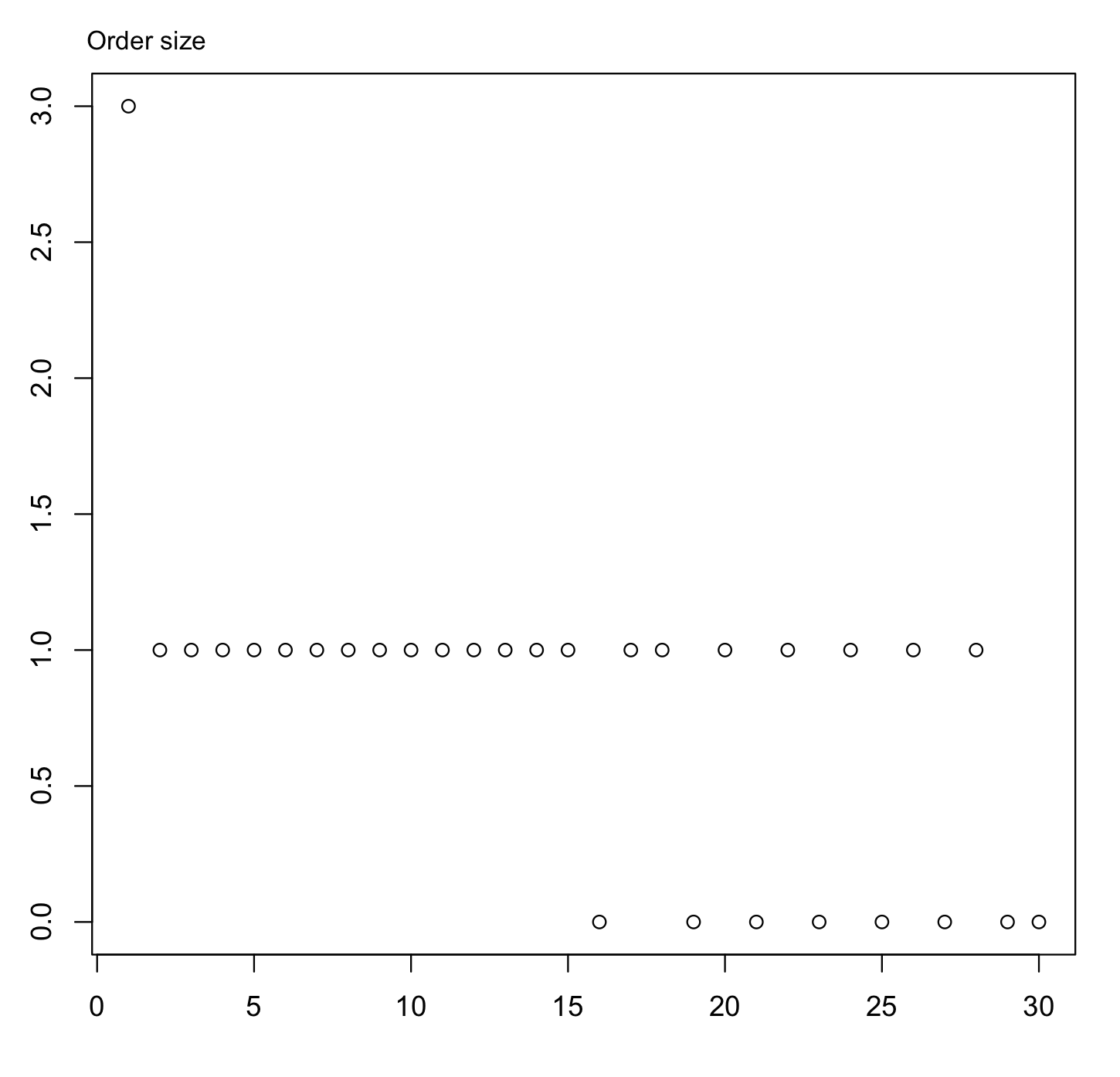} \includegraphics[scale=0.29]{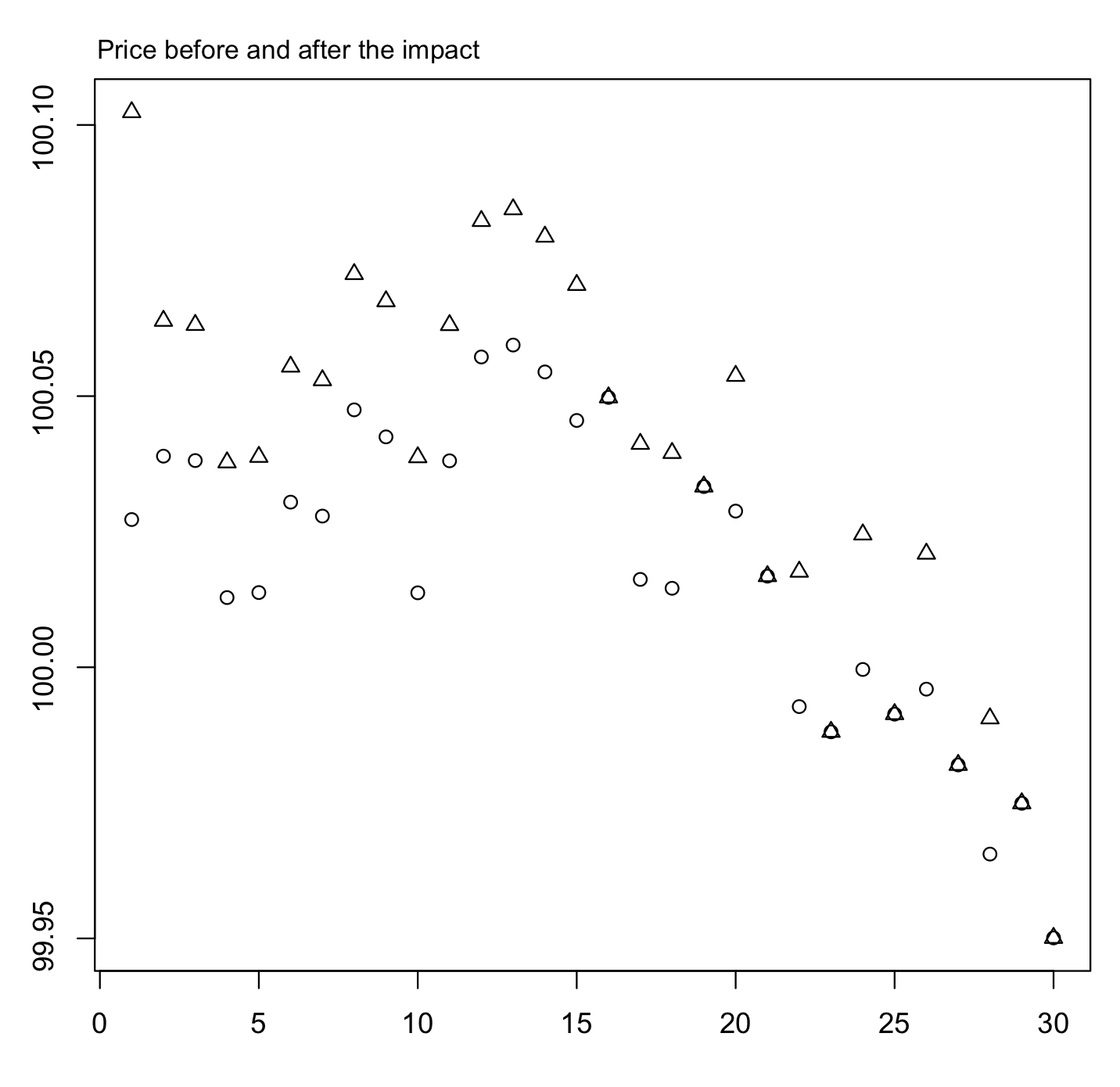}
 \end{center}
 \begin{center}
\includegraphics[scale=0.29]{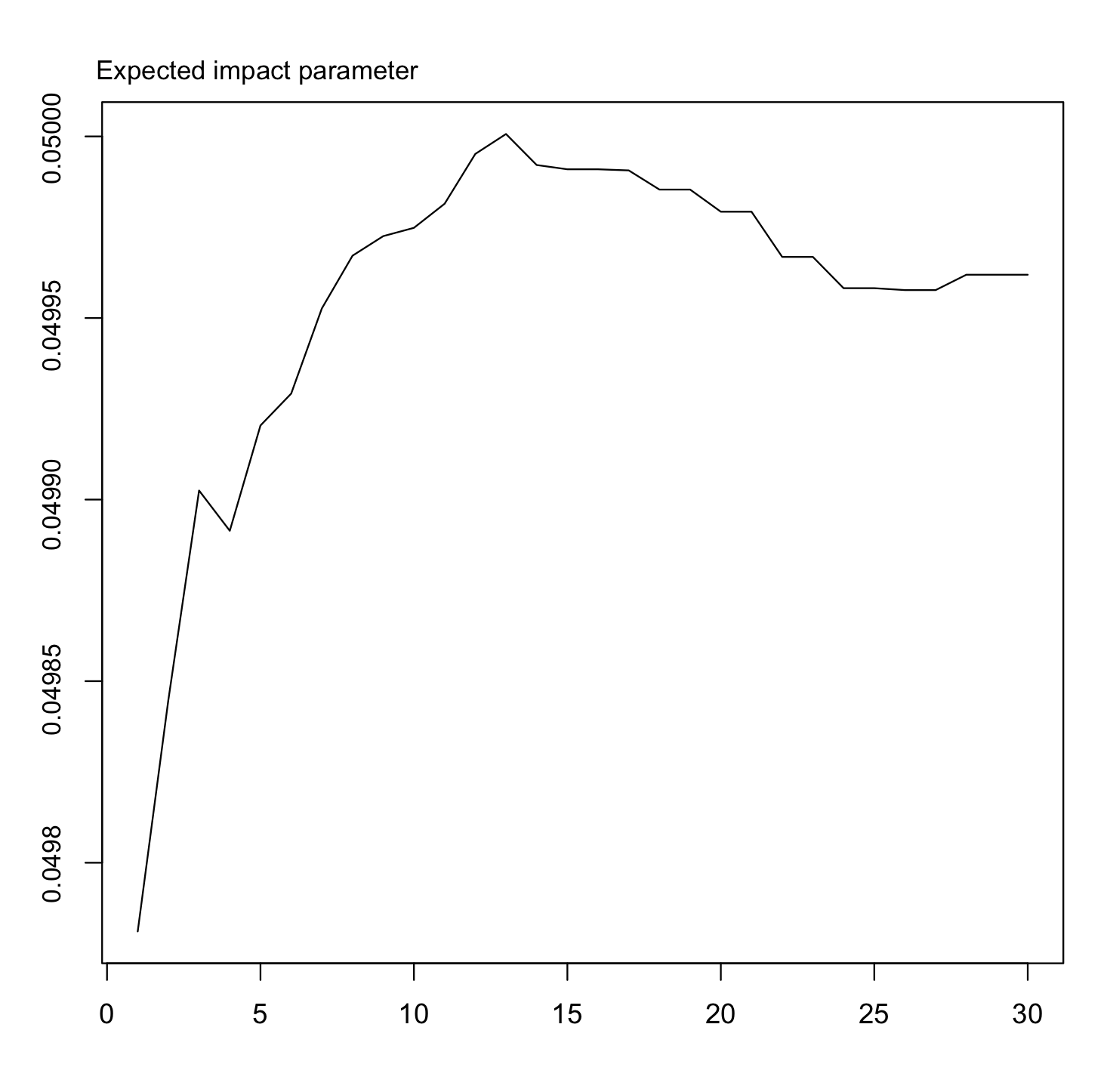} \includegraphics[scale=0.29]{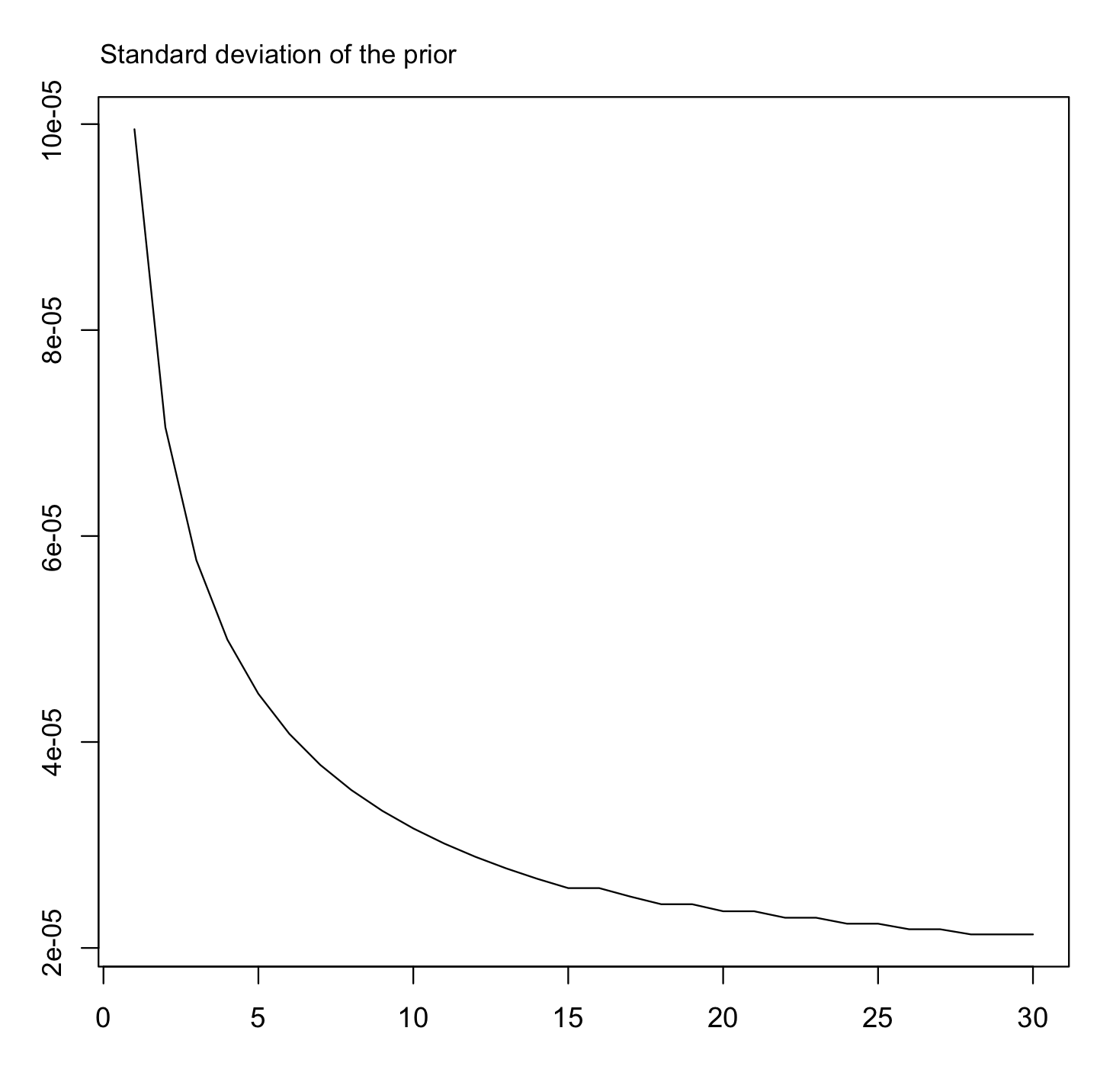}
 \end{center}
 \captionof{figure}{Evolution of $\beta$ (top left), price before (circles) and after (triangles) the impact (top right), $m_{\upsilon}$ (bottom left), $\sigma_{\upsilon}$ (bottom right), with time in second. The true value of $\upsilon$ is $5.10^{-2}$. $x$-axis: time in seconds.}
\label{resilience2D}
\end{figure}

\clearpage 
  %%%%%%%%%%%%%%%%%%
 
 Let us now  consider the case  $\rho=0$, i.e.~without dynamic resilience, with a trading period of $60$ seconds and  $N=50$.
In Figure \ref{fig: gA aggressif}, we provide the optimal policy (number of traded shares) in terms of the number $X^{3}$ of already {traded shares} and the prior's mean parameter $m_{\upsilon}$ for different times.  Not surprisingly the algorithm is more aggressive as the prior's mean {decreases} and the remaining number of {shares} to buy {increases}. It is rather stable in time (compare $t=0$s with $t=30$s) up to the end where it is forced to accelerate to avoid a large final impact cost. It is also much more aggressive compared to the case $\rho>0$ presented above: we can no more make profit of the decrease of the resilience term $X^{4}$, and there is no reason to wait. 
  \begin{figure}
 \begin{center}
 \includegraphics[scale=0.4]{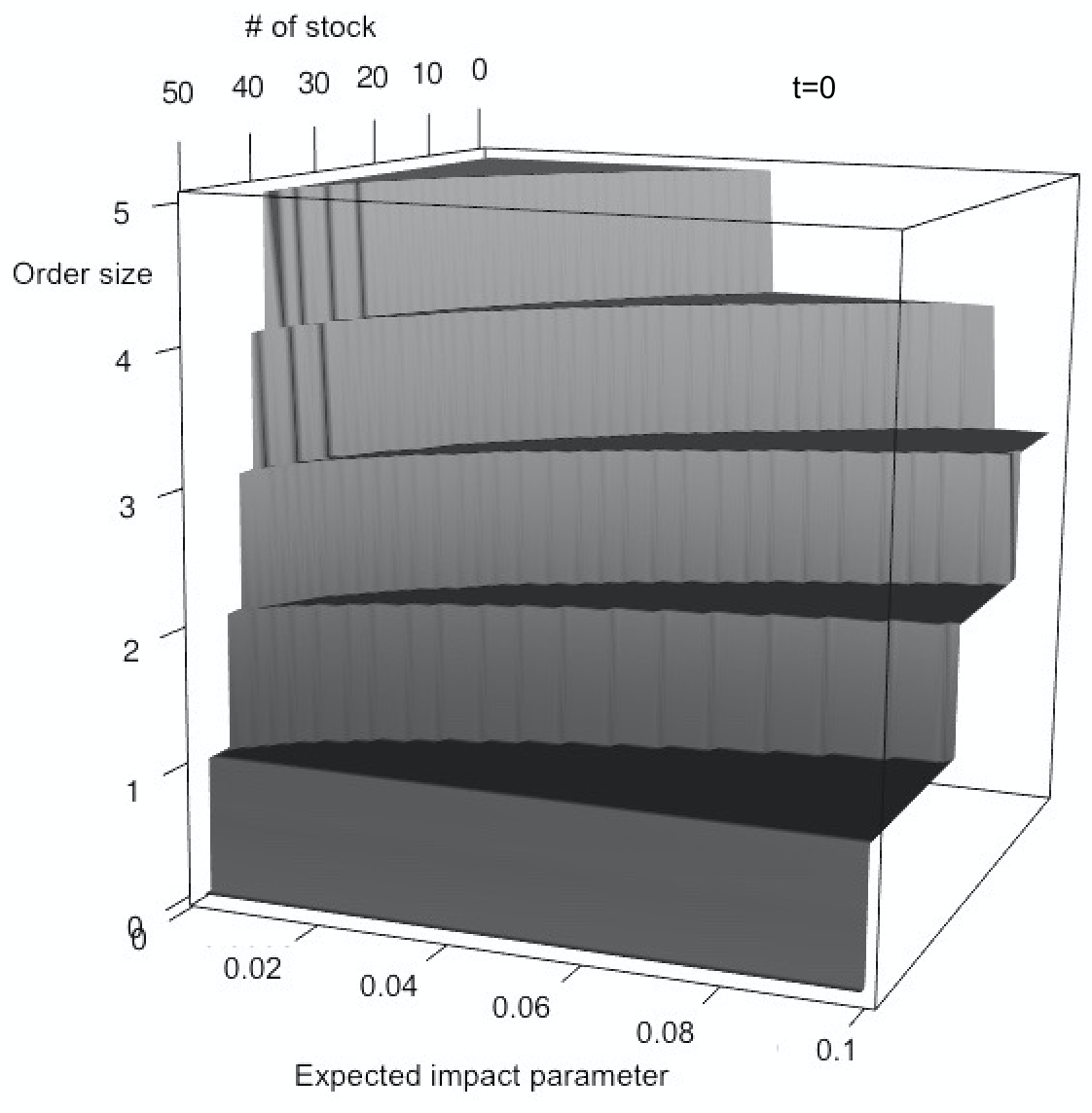}
 \end{center}
 \includegraphics[scale=0.4]{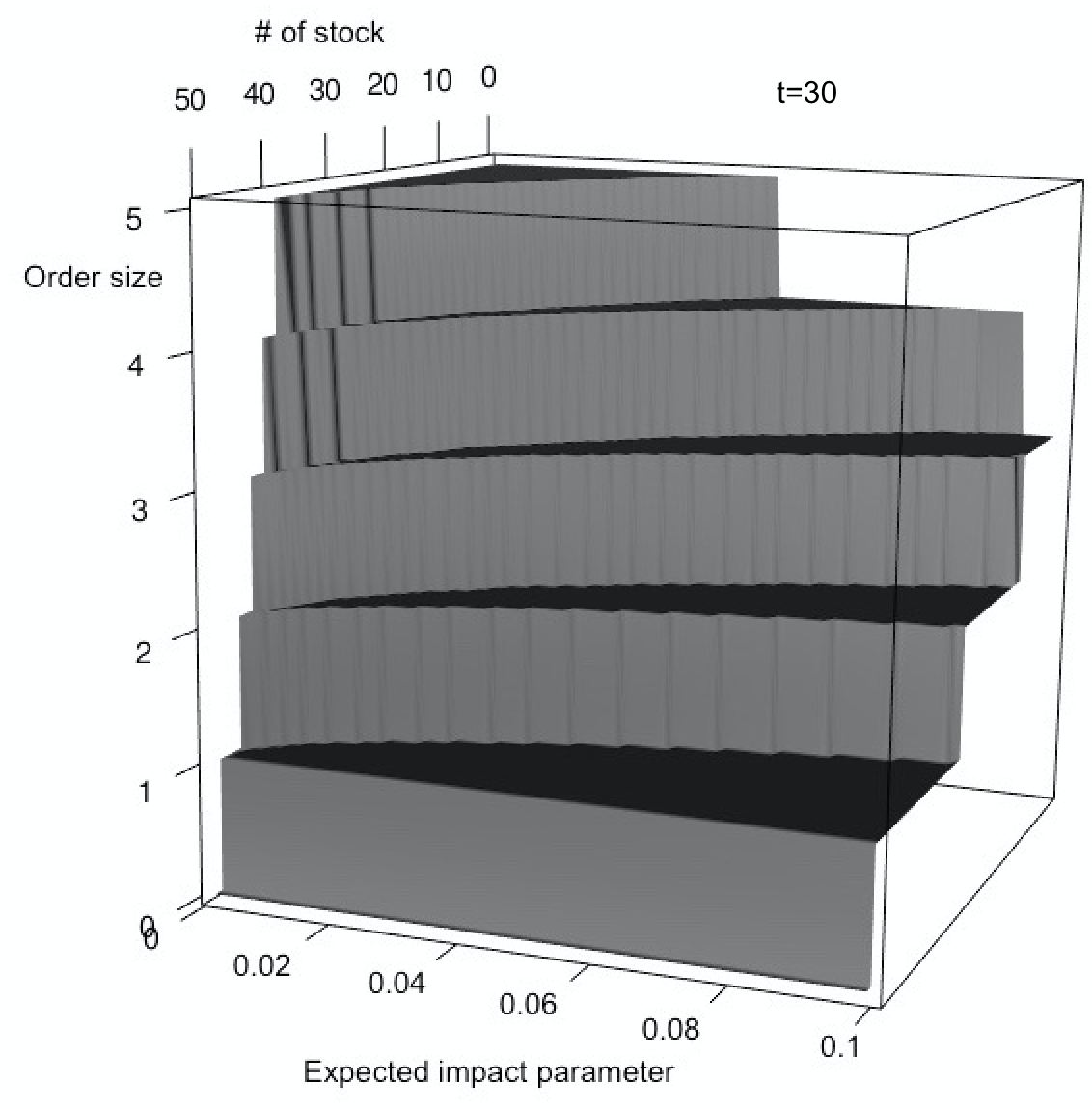} \includegraphics[scale=0.4]{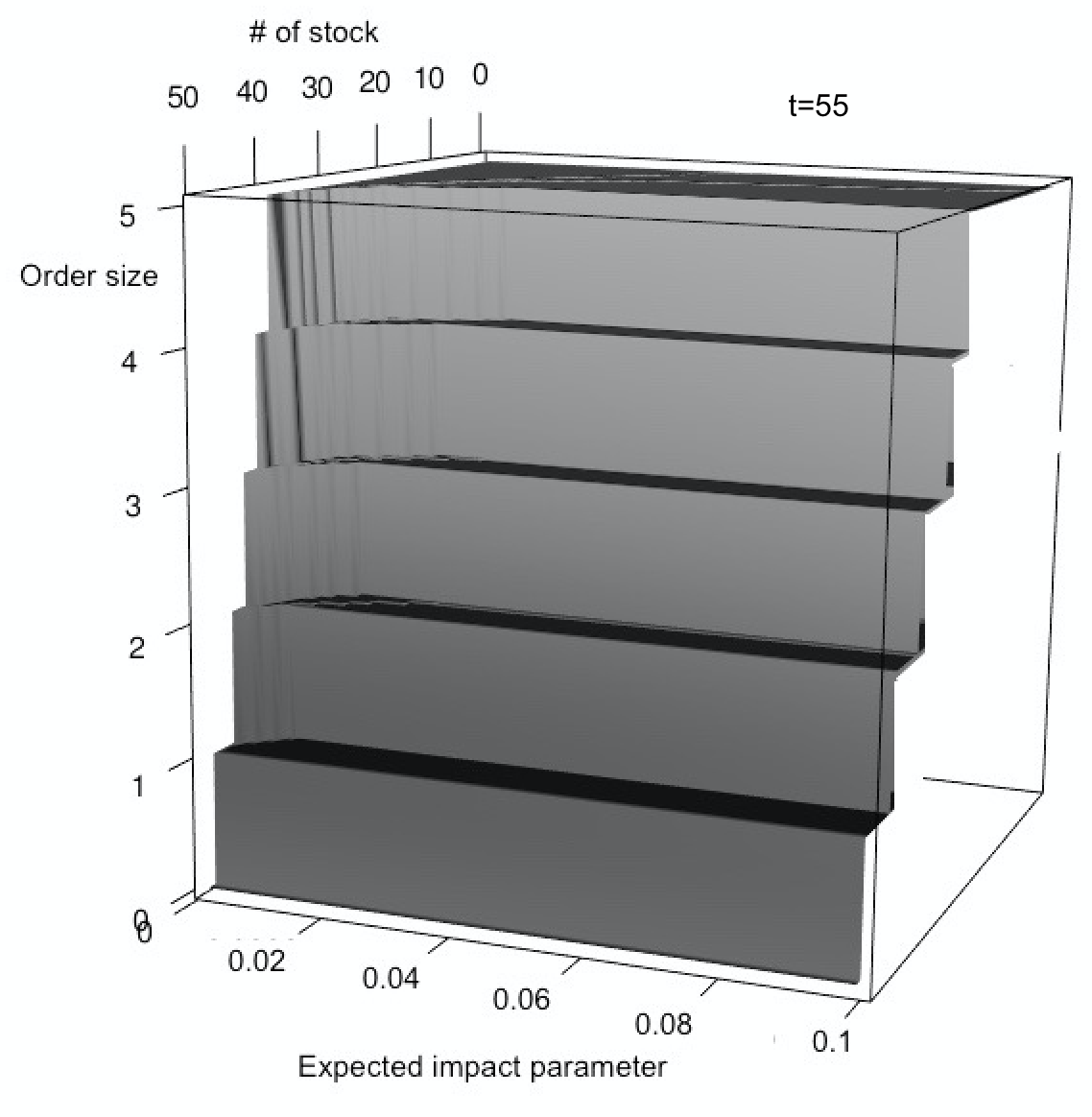}
 \captionof{figure}{Evolution of $\alpha$ in terms of $(m_{\upsilon},X^{3})$ at time $0$s (top), $30$s (left) and $55$s (right), for $\sigma_{\upsilon}=5.10^{{-4}}$.}
\label{fig: gA aggressif}
\end{figure}

In Figure \ref{fig: spath aggressif}, we provide a simulated path of $(X,\alpha,m_{\upsilon},\sigma_{\upsilon})$ that shows how the prior on the unknown coefficients $\upsilon$ can adapt to changing market conditions.  The red dashed lines and circles correspond to the same path of Brownian motion and the same realized noises $(\epsilon_{i})_{i\ge 1}$ as the black solid lines and crosses, but the true parameter is changed from $5.10^{-2}$ to $5.10^{-4}$ after $5$ seconds.  It is more aggressive quite quickly after the shock as the prior adapts to the new small level of impact. Note that the total number of {shares} is bought slightly before $30$s, so that the prior do not change anymore after this date. 
 \begin{figure}
 \begin{center}
\hspace{-1cm} \includegraphics[scale=0.3]{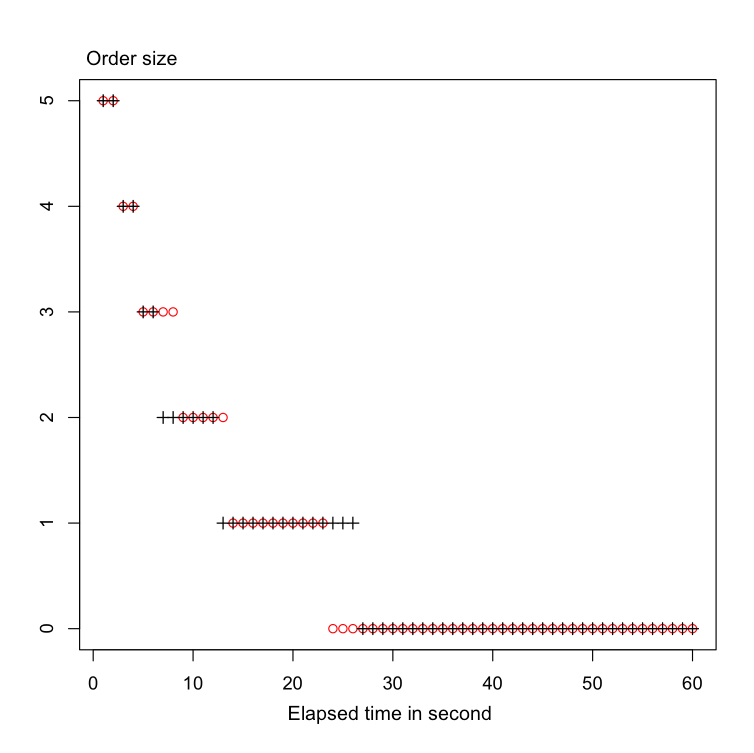}
 \end{center}
\hspace{-1cm}\includegraphics[scale=0.3]{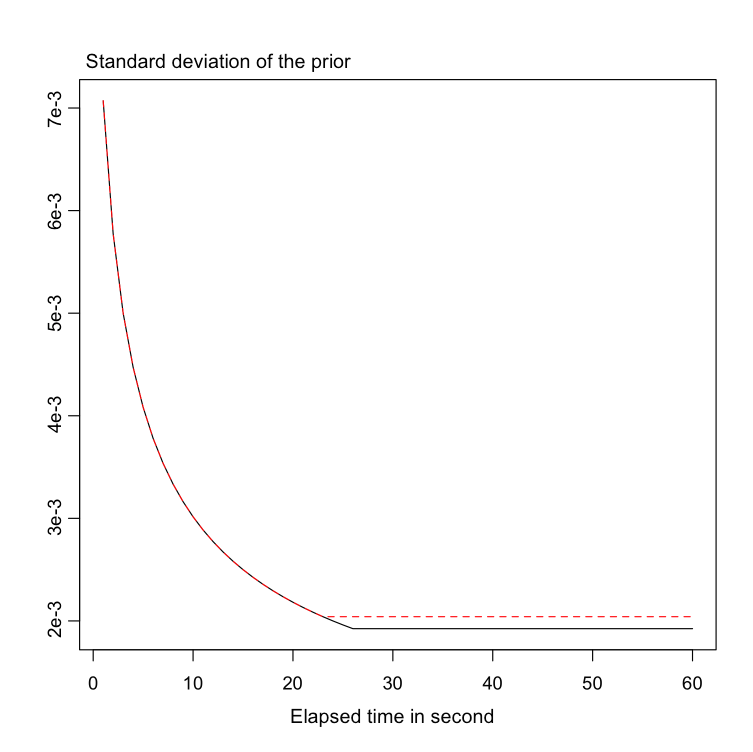}\includegraphics[scale=0.3]{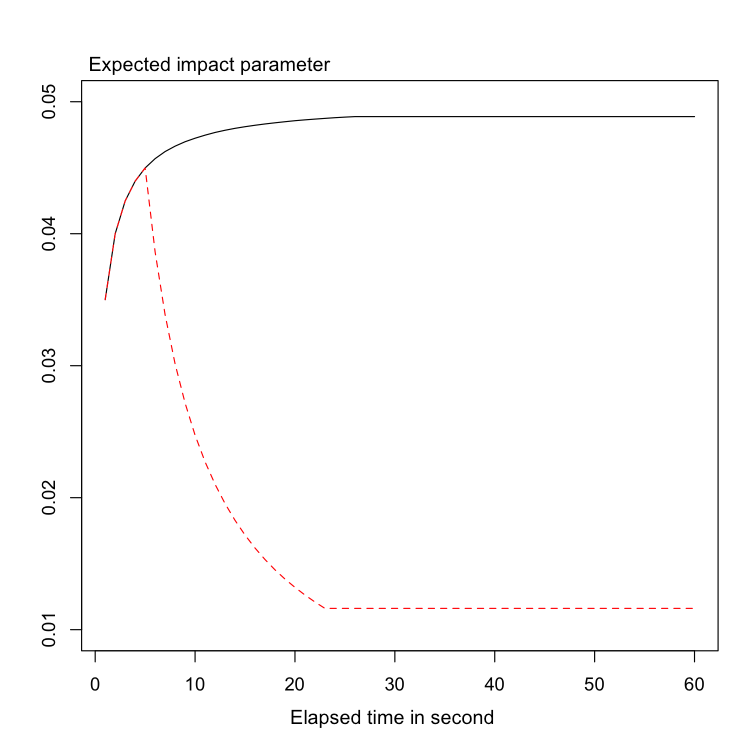}
 \captionof{figure}{Evolution of $\alpha$ (top), $\sigma_{\upsilon}$ (left), $m_{\upsilon}$ (right) with time. Black crosses and black solid lines: the true value of $\upsilon$ is $5.10^{-2}$. Red circles and red dashed lines: the true value of $\upsilon$ is $5.10^{-2}$ for the first $5$ seconds, and then jumps to $5.10^{-4}$.}
\label{fig: spath aggressif}
\end{figure}

In Figure \ref{fig: ln(v)}, we plot the log of the value function minus the {cost} $5.10^{3}$ of buying the {total shares} without impact{ (similar to the {\sl implementation shortfall})}, in terms of the different quantities of interest. 
\begin{figure}
\begin{center}
  \includegraphics[scale=0.35]{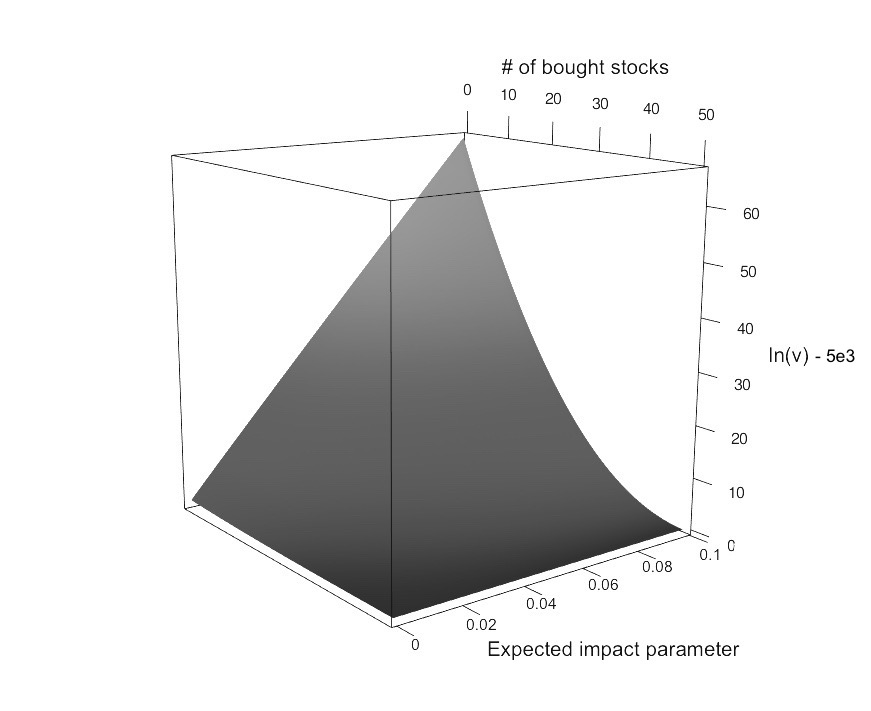}\\\includegraphics[scale=0.32]{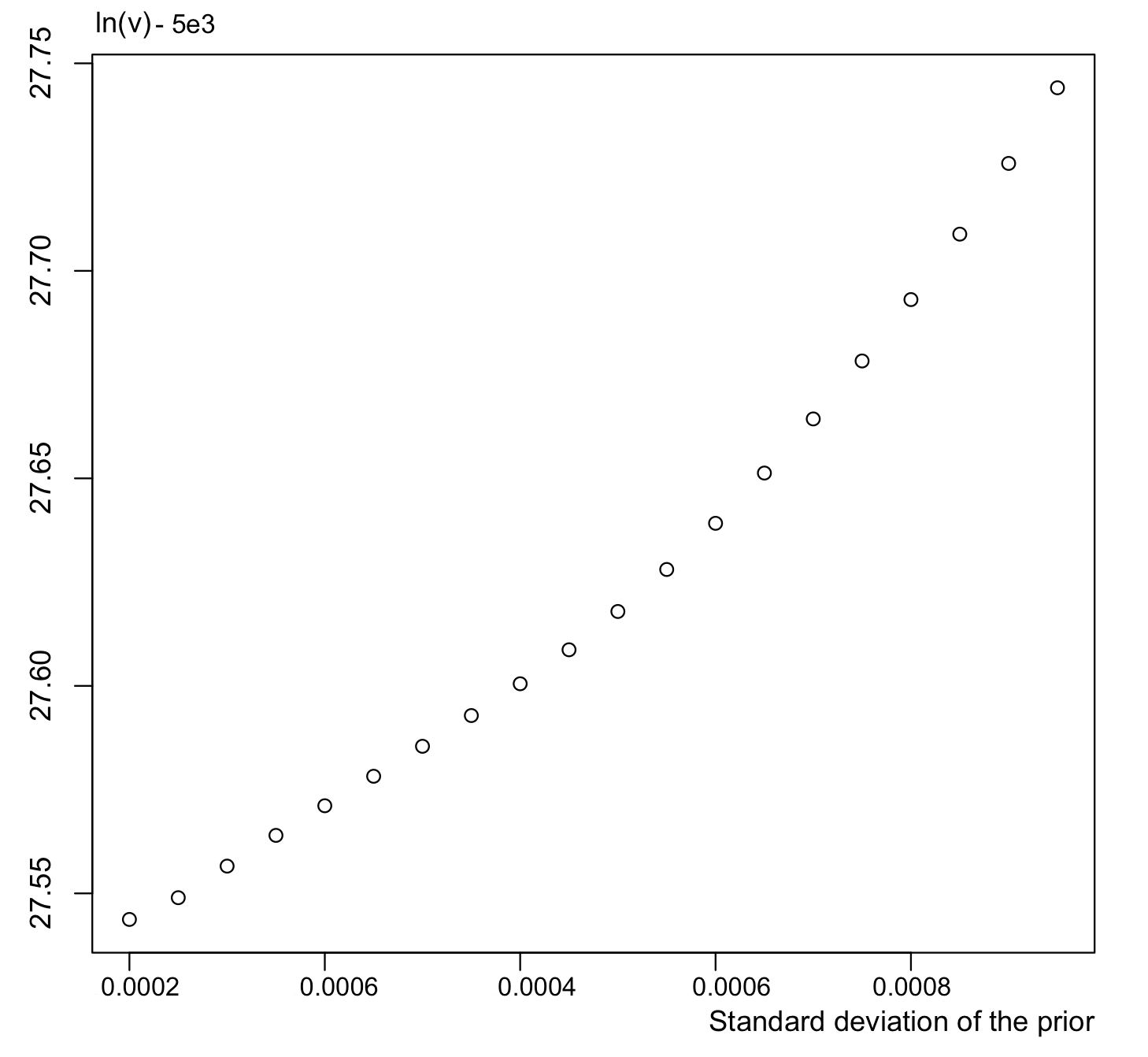}
 \captionof{figure}{Top: $\ln({\rm v})-5.10^{3}$ in terms of $(x^{3},m_{\upsilon})$ for $\sigma_{\upsilon}=5.10^{-4}$ and $t=0$. Bottom: $\ln({\rm v})-5.10^{3}$ in terms of $\sigma_{\upsilon}$ for $(x^{3},m_{\upsilon})=(0,2.10^{-2})$ at $t=0$.}
\label{fig: ln(v)}
\end{center}
\end{figure}
\newpage

    \subsection{Random execution times: application to strategies using limit-orders}

In this section, we consider a limit-order trading model. $X^{1}$ now represents a mid-price (of reference) and, between two trades, has the dynamic

\begin{equation}
dX_{t}^{1} = \sigma dW_{t}^{1}.%, t \not\in \{\tau_{i}, i \geq 1\},
\end{equation}

An order is of the form $(\ell, \beta)$ in which $\ell$ is the maximal time we are ready to wait before being executed, while $\beta$ is the price at which the limit order is sent\footnote{Dark pool strategies could be considered similarly, in this case, $\beta$ would  rather describe the choice of the trading plateform}. For simplicity, each order corresponds to buying one share.

We assume that the time $\theta$ it takes to be executed follows an exponential distribution of parameter $\rho(\upsilon, X_{\tau}^{1} - \beta)$, given the information at time $\tau$. One can send a new order only after $\vartheta:=\tau + \ell \wedge \theta$. 

Hence, given a flow of orders $\phi = (\tau_{i}, \ell_{i}, \beta_{i})_{i \geq 1}$, the number $X^{3}$ of shares bought evolves according to
	\[
		\begin{aligned}
			X^3 &= X_{\vartheta_{i}}^{3} \text{ \ on } [\vartheta_{i}, \tau_{i+1}) \\
			X^{3}_{\vartheta_{i}} &= X_{\tau_{i}-}^{3} + \mathbf{1}_{\left\{\theta_{i} \leq \ell_{i}\right\}},
		\end{aligned}
	\]
in which $\vartheta_{i}:=\tau_{i} + \ell_{i} \wedge \theta_{i}$. Each $\theta_{i}$ follows an exponential distribution of parameter $\rho(\upsilon, X_{\tau_{i}}^{1} - \beta_{i})$ given $\Fc^{z,m,\phi}_{\tau_{i}-}$. As in the previous model, $X^{3}$ is restricted to $\{0,\ldots,N\}$.   The total cost $X^{2}$ of buying the shares has the dynamics
	\[
		\begin{aligned}
			X^2 &= X_{\vartheta_{i}}^{2} \text{ on } [\vartheta_{i}, \tau_{i+1}) \\
			X^{2}_{\vartheta_{i}} &= X_{\tau_{i}-}^{2} + \beta_{i}\mathbf{1}_{\left\{\theta_{i} \leq \ell_{i}\right\}}.
		\end{aligned}
	\]

We want to maximize
	\[
		-\mathbb{E}\left[e^{X^{2}_{\T[\phi]} + 1.02(N-X^{3}_{\T[\phi]}) + \frac{5.10^{2}}{2}(N-X^{3}_{\T[\phi]})^{2}}\wedge C\right],
	\]
in which $1.02$ is the best  ask (kept constant) and $5.10^{2}$ is an impact coefficient. This corresponds to the cost of liquidating instantaneously the remaining shares $(N-x^{3})^{+}$ at $T$.  
This model is a version of \cite{avellaneda2008high}, \cite{gueant2012optimal}, \cite{ho1981optimal}, see also \cite{gueant2013dealing}.
\vs2

Direct computations show that the prior process $M$   evolves according to
	\[
		\begin{aligned}
			M &= M_{\vartheta_{i}} \text{ \ on } [\vartheta_{i}, \tau_{i+1} ) \\
			M_{\vartheta_{i}} &= \mathfrak{M}_{1}(M_{\tau_{i}-}; Z_{\vartheta_{i}}, Z_{\tau_{i}-},   \alpha_{i}) \mathbf{1}_{\{\theta_{i} \leq \ell_{i}\}} 
				+ \mathfrak{M}_{2}(M_{\tau_{i}-}; Z_{\vartheta_{i}},  Z_{\tau_{i}-},  \alpha_{i}) \mathbf{1}_{\{\theta_{i} > \ell_{i}\}}
		\end{aligned}
	\]
in which
	\[
		\mathfrak{M}_{1}(m; t' ,x', t, x, l, b)[B] := \frac{\int_{B}\rho(u, x^{1}-b)e^{-\rho(u, x^{1}-b)t'}dm(u)}{\int_{\mathbb{R}^{+}}\rho(u, x^{1}-b)e^{-\rho(u, x^{1}-b)t'}dm(u)}
	\]
and
	\[
		\mathfrak{M}_{2}(m; t',x', t, x, l, b)[B] := \frac{\int_{B}e^{-\rho(u, x^{1}-b)l}dm(u)}{\int_{\mathbb{R}^{+}}e^{-\rho(u, x^{1}-b)l}dm(u)}
	\]
for all Borel set $B$.

In the case where $\Mb$ is the convex hull of a finite number of Dirac masses, then the weights associated to $M$ can be computed explicitly. 
\vs2

Here again, the map $\Psi(t,x,m)=N-x^{3}$ satisfies the conditions provided in \cite{BBD16} to ensure that Assumption \ref{ass: comp} holds.

\vs2
We now consider a numerical illustration. We take $C=10^{200}$. The time horizon is $T=15$ minutes. To simplify, we fix the reference mid-price to be $X^{1} \equiv 1$ (i.e. $\sigma=0$) and restrict to $\ell =  1$, i.e. an order is sent each minute.   We take $N=10$. One can send limit buy orders in the range $B := \{0.90, 0.92, 0.94, 0.96, 0.98\}$.

As for the intensity of the execution time, we use an exponential form as in \cite{gueant2012optimal}: $\rho(u, x^{1} - b) = \lambda(u)e^{-20(0.98-b)}$ in which $\lambda(u) = -\ln(1-u)$. This means that the probability to be executed at the price 0.98 within one minute is $u$. Orders are sent each minute, but we use a finner time grid in order to take into account that it can be executed before this maximal time-length. The original prior is supported by two Dirac masses at $u=0.3$ and $u=0.8$. The corresponding probabilities of being executed within one minute are plotted in Figure \ref{prob_exec}.

  \begin{figure}
\begin{center}
\includegraphics[scale=0.4]{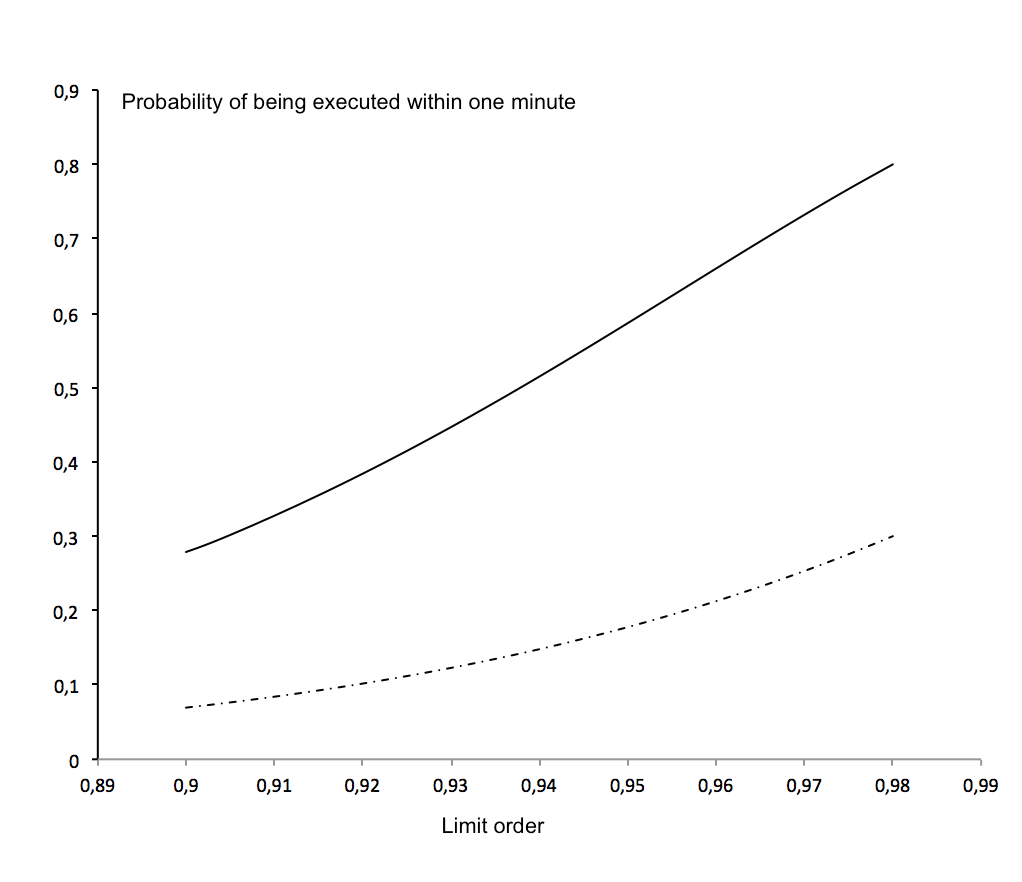}
 \captionof{figure}{Solid: $u=0.8$. Dashed: $u=0.3$}
\label{prob_exec}
\end{center}
\end{figure}

Our time step corresponds to 15 seconds, so that every 15 seconds the controller can launch a new order if the previous one has been executed before the maximal 1 minute time-length. In Figure \ref{log_diff}, we plot the difference, in logarithms, between the value functions obtained in the latter case and for a time step of 1 minute (in which case a new order cannot be launched before one minute). Clearly, the possibility of launching new orders in advance is an advantage.

   \begin{figure}[!h]
\begin{center}
\includegraphics[scale=0.31]{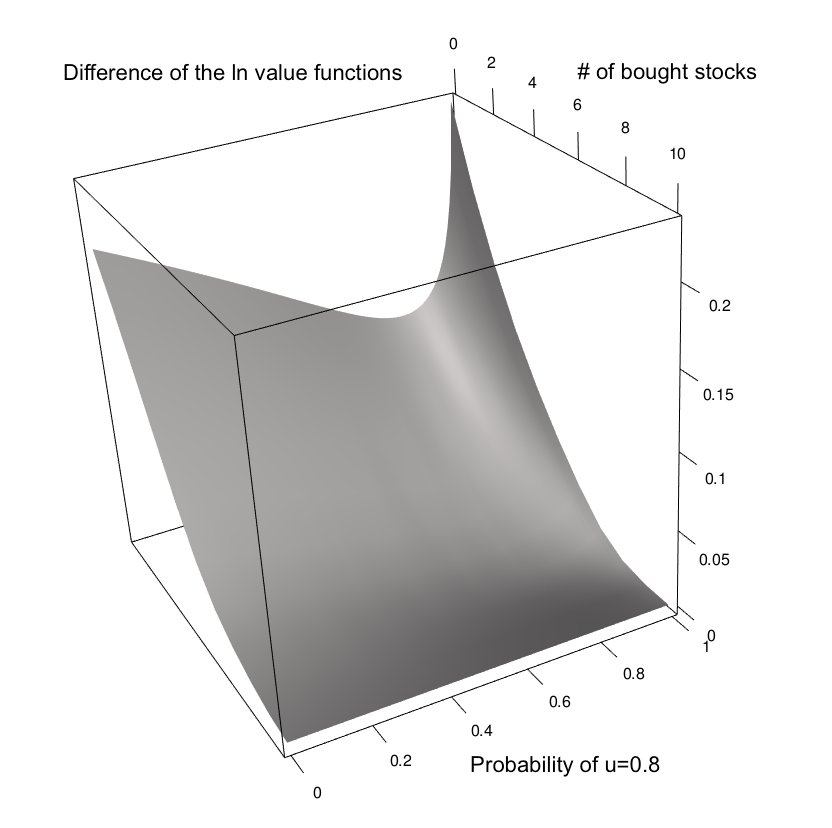}
 \captionof{figure}{}
\label{log_diff}
\end{center}
\end{figure}

In Figure \ref{fig:A dark}, we plot the optimal policy at time $t=0$ and $t=7.5$ minutes. As expected,  the algorithm is more aggressive when the probability of having $\upsilon=0.8$ is higher.

In Figure \ref{fig:simul dark 1}, we plot a simulated path. The red and black lines and points correspond to the same realization of the random variables at hand, but for different values of the real value of $\upsilon$. Black corresponds to the most favorable case $\upsilon=0.8$, while red corresponds to $\upsilon=0.8$ for the first $7.5$ minutes and $\upsilon=0.3$ for the remaining time. The initial prior is $\P[\upsilon=0.8]=9\%$. Again, the algorithm adapts pretty well to this shock on  the true parameter.  We also see that it is more aggressive when the prior probability of being in the favorable case is high.  \newpage

%%%%%%%%%%%%%%%%%%%%%%%%%%%%%%%
%
  \begin{figure}
\begin{center}
\includegraphics[scale=0.4]{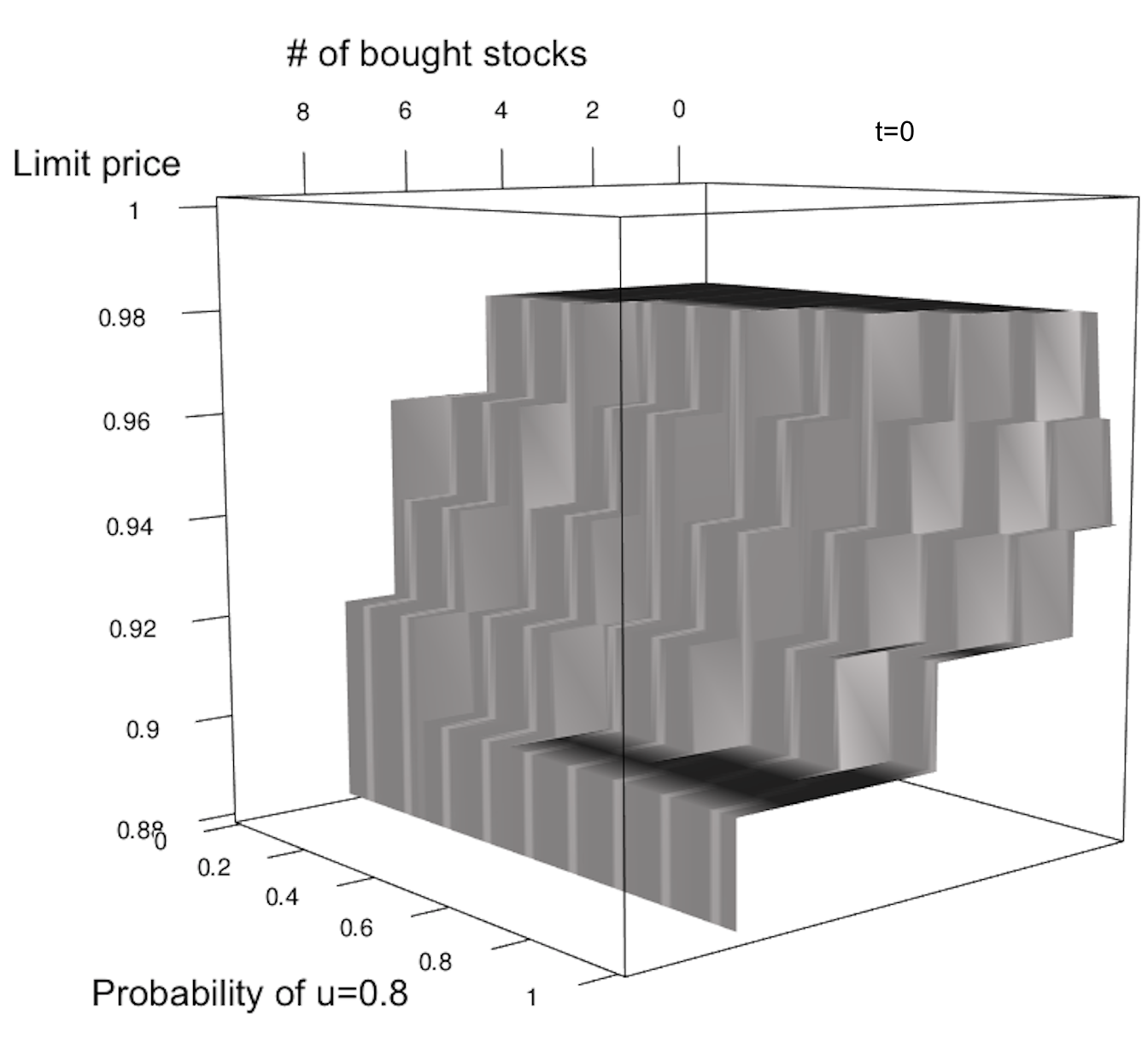}\\
~~~\\

\includegraphics[scale=0.4]{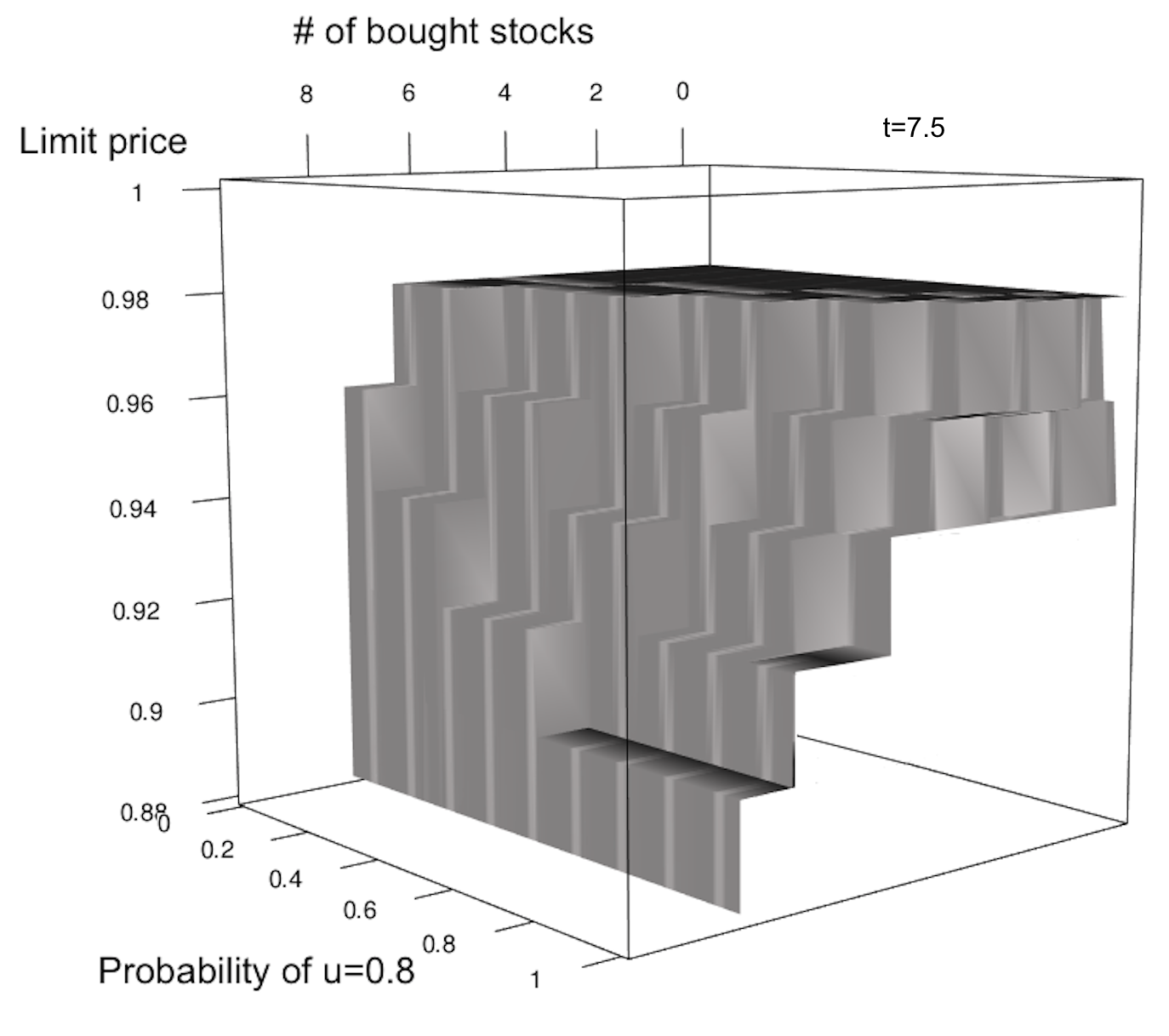}
 \captionof{figure}{Top: $t=0$. Bottom: $t=7.5$ minutes}
\label{fig:A dark}
\end{center}
\end{figure}

  \begin{figure}
\includegraphics[scale=0.29]{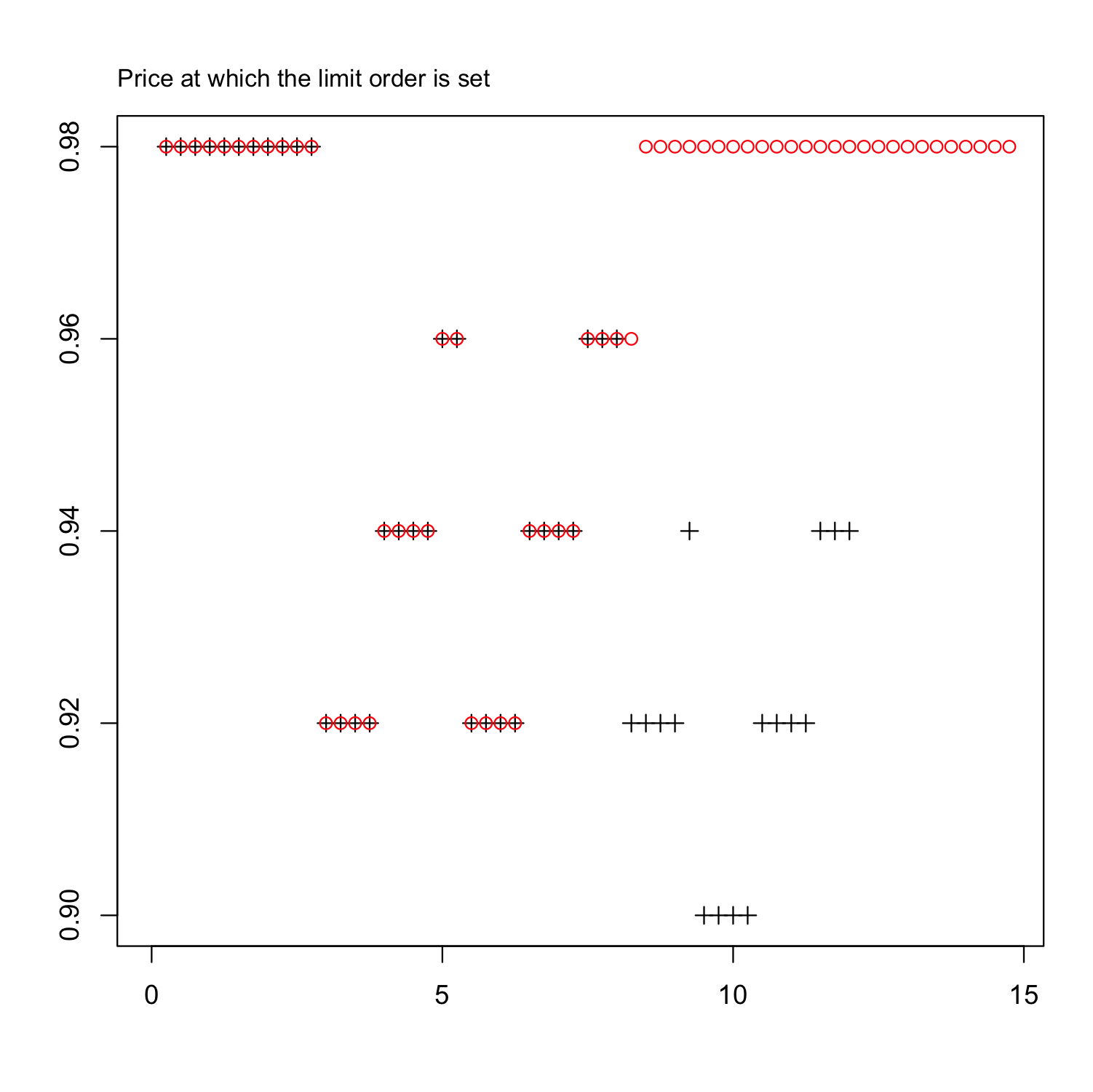}\includegraphics[scale=0.29]{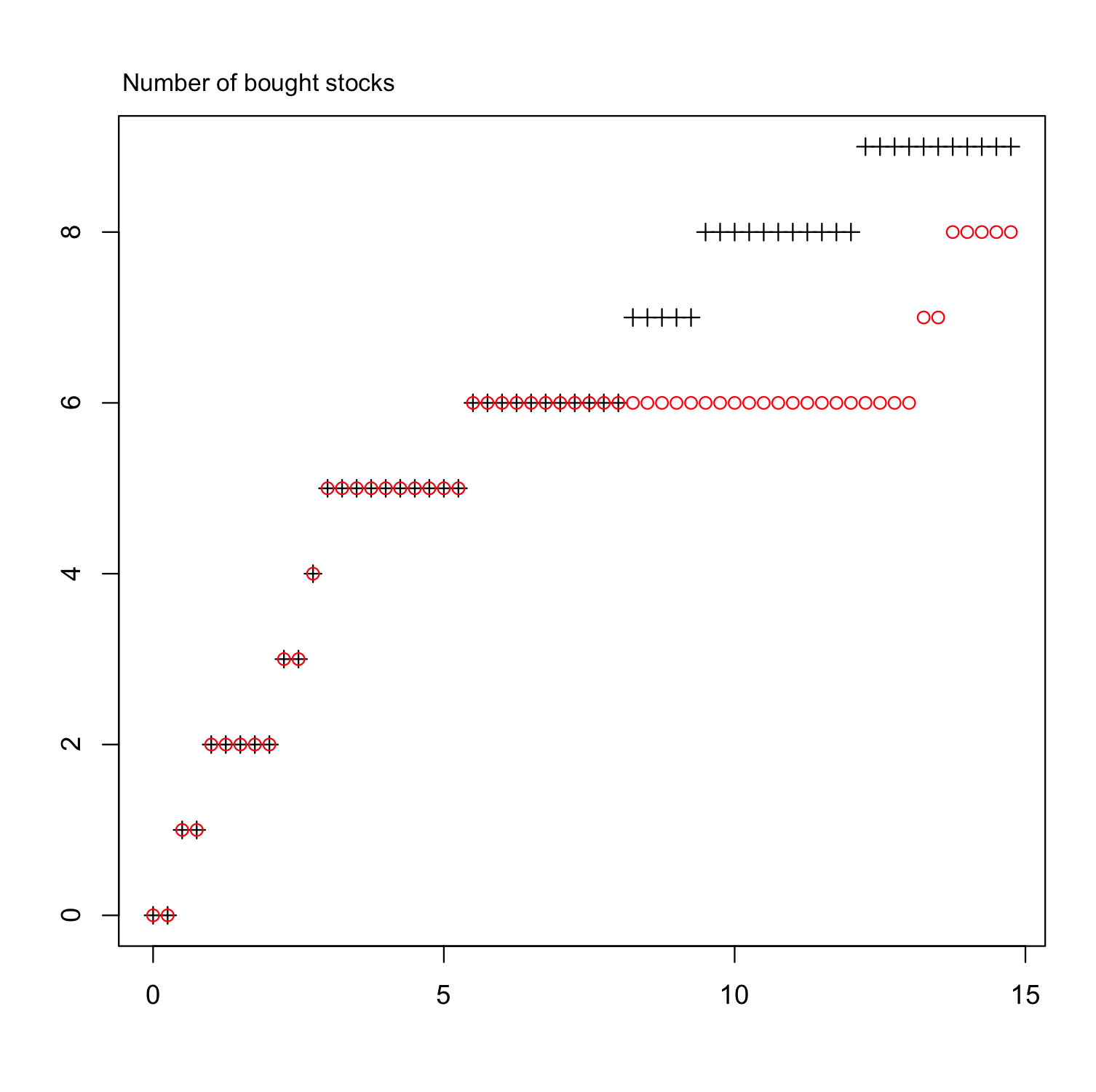}\\

\includegraphics[scale=0.29]{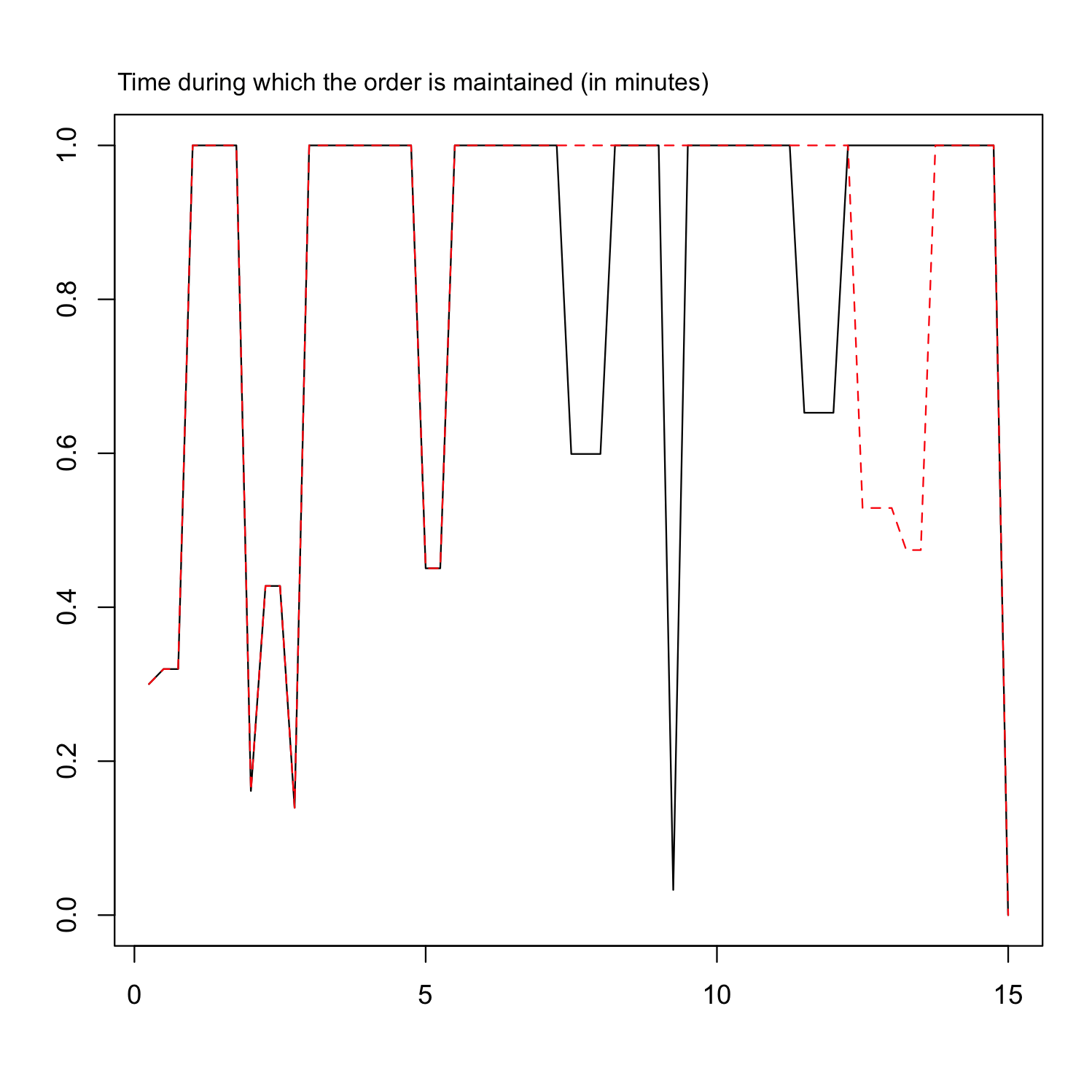}\includegraphics[scale=0.29]{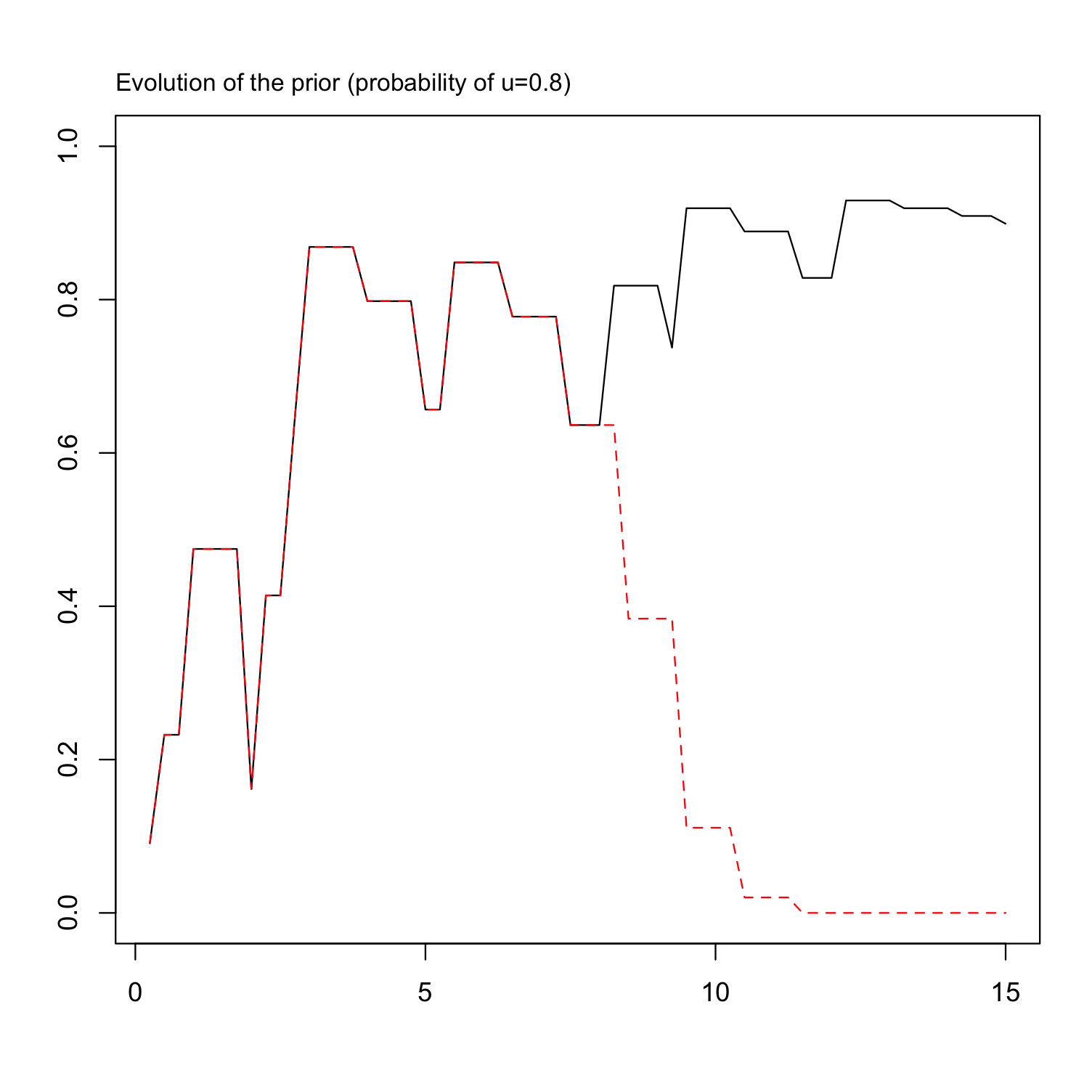}
 \captionof{figure}{Black crosses and solid lines: $\upsilon=0.8$. Red circles and dashed lines: $\upsilon=0.8$ before $t=7.5$ minutes and $\upsilon=0.3$ after. $x$-axis$=$ time in minutes.}
\label{fig:simul dark 1}.
\end{figure}

 \clearpage
%%%%%%%%%%%%%%%%%%%%%%%%%%%%%%%%%%%%%%%%%%%%%%%%%%
 %%%%%%%%%%%%%%%%%%%%%%%%%%%%%%%%%%%%%%%%%%%%%%%%%%
 
 %%%%%%%%%%%%%%%%%%%%%%%%%%%%%%%
 \bibliographystyle{plain}

  \end{document}